\documentclass[lettersize,journal]{IEEEtran}
\IEEEoverridecommandlockouts
\usepackage{cite}
\usepackage{amsmath,amssymb,amsfonts}
\usepackage[amsmath,thmmarks]{ntheorem}
\usepackage{algorithmic}
\usepackage{graphicx} 
\usepackage{epsfig} 
\usepackage{epstopdf} 
\usepackage{textcomp}
\usepackage{xcolor}
\usepackage{multirow}

\newenvironment{proof}{{\indent \indent \it Proof:}}{\hfill $\blacksquare$\par}

{\theoremheaderfont{\itshape}
\theorembodyfont{\upshape}
\theoremseparator{:}
\newtheorem{proposition}{\hspace{1em}Proposition}}

{\theoremheaderfont{\itshape}
\theorembodyfont{\upshape}
\theoremseparator{:}
}

{\theoremheaderfont{\itshape}
\theorembodyfont{\upshape}
\theoremseparator{:}
}

{\theoremheaderfont{\itshape}
\theorembodyfont{\upshape}
\theoremseparator{:}
}

{\theoremheaderfont{\itshape}
\theorembodyfont{\upshape}
\theoremseparator{:}
\newtheorem{remark}{\hspace{1em}Remark}}
\usepackage{indentfirst}
\usepackage{array}
\usepackage{bm}
\usepackage[caption=false,font=footnotesize,labelfont=rm,textfont=rm]{subfig}
\usepackage{textcomp}
\usepackage{stfloats}
\usepackage{url}
\usepackage{verbatim}
\usepackage{graphicx}
\usepackage{caption}
\usepackage{amssymb}
\usepackage{float}
\usepackage{balance}
\usepackage{diagbox}
\usepackage{booktabs}
\usepackage{subcaption}

\allowdisplaybreaks[4]

\makeatletter

\newcommand{\Rmnum}[1]{\expandafter\@slowromancap\romannumeral #1@}
\makeatother

\def\BibTeX{{\rm B\kern-.05em{\sc i\kern-.025em b}\kern-.08em
    T\kern-.1667em\lower.7ex\hbox{E}\kern-.125emX}}

\usepackage {setspace}
\makeatletter
\newcommand\notsotiny{\@setfontsize\notsotiny\@vipt\@viipt}
\makeatother

\usepackage[explicit]{titlesec}
 
\titlespacing*{\section}{0pt}{1.2ex plus .0ex minus .0ex}{.3ex plus .0ex}
\titlespacing*{\subsection}{0pt}{1.2ex plus .0ex minus .0ex}{.3ex plus .0ex}

\begin{document}
\title{Performance Analysis of Wireless-Powered Pinching Antenna Systems 
}

\author{Kunrui Cao, Jingyu Chen, Panagiotis D. Diamantoulakis,~\IEEEmembership{Senior Member,~IEEE}, Lei Zhou,\\ Xingwang Li,~\IEEEmembership{Senior Member,~IEEE}, Yuanwei Liu,~\IEEEmembership{Fellow,~IEEE}, and George K. Karagiannidis,~\IEEEmembership{Fellow,~IEEE}

\thanks{
    Kunrui Cao, Jingyu Chen, and Lei Zhou are with Information Support Force Engineering University, Wuhan 430035, China, and also with the School of Information and Communications, National University of Defense Technology, Wuhan 430035, China (e-mail: krcao@nudt.edu.cn; chenjingyu@nudt.edu.cn; cat\_radar@163.com).
    
    Panagiotis D. Diamantoulakis, and George K. Karagiannidis are with the Department of Electrical and Computer Engineering, Aristotle University of Thessaloniki, Thessaloniki 54124, Greece (padiaman@auth.gr; geokarag@auth.gr).

    Xingwang Li is with the School of Physics and Electronic Information Engineering, Henan Polytechnic University, Jiaozuo 454099, China, and also with the National Mobile Communications Research Laboratory, Southeast University, Nanjing 210096, China (e-mail:lixingwangbupt@gmail.com).

    Yuanwei Liu is with the Department of Electrical and Electronic Engineering, the University of Hong Kong, Hong Kong, China (e-mail: yuanwei@hku.hk).
    }

}
\maketitle

\begin{abstract}
    Pinching antenna system (PAS) serves as a groundbreaking paradigm that enhances wireless communications by flexibly adjusting the position of pinching antenna (PA) and establishing a strong line-of-sight (LoS) link, thereby reducing the free-space path loss.
    This paper introduces the concept of wireless-powered PAS, and investigates the reliability of wireless-powered PAS to explore the advantages of PA in improving the performance of wireless-powered communication (WPC) system.
    In addition, we derive the closed-form expressions of outage probability and ergodic rate for the practical lossy waveguide case and ideal lossless waveguide case, respectively, and analyze the optimal deployment of waveguides and user to provide valuable insights for guiding their deployments.
    The results show that an increase in the absorption coefficient and in the dimensions of the user area leads to higher in-waveguide and free-space propagation losses, respectively, which in turn increase the outage probability and reduce the ergodic rate of the wireless-powered PAS.
    However, the performance of wireless-powered PAS is severely affected by the absorption coefficient and the waveguide length, e.g., under conditions of high absorption coefficient and long waveguide, the outage probability of wireless-powered PAS is even worse than that of traditional WPC system.
    While the ergodic rate of wireless-powered PAS is better than that of traditional WPC system under conditions of high absorption coefficient and long waveguide.
    Interestingly, the wireless-powered PAS has the optimal time allocation factor and optimal distance between power station (PS) and access point (AP) to minimize the outage probability or maximize the ergodic rate.
    Moreover, the system performance of PS and AP separated at the optimal distance between PS and AP is superior to that of PS and AP integrated into a hybrid access point.
\end{abstract}

\begin{IEEEkeywords}
    Pinching antenna, wireless-powered communications, performance analysis.
\end{IEEEkeywords}

\section{Introduction}
With the large-scale commercialization of the fifth generation (5G), the global industry and academia have begun research and exploration into the next generation of mobile communication technology, namely the sixth generation (6G).
Compared to 5G, 6G will aim for faster transmission rates, larger-scale connectivity, broader coverage, ultra-low latency, ultra-high precision positioning, integrated communication and sensing, greater intelligence, enhanced security, and improved interoperability \cite{10054381}.
In June 2021, the white paper on 6G overall vision and potential key technologies published by the IMT-2030 Promotion Group envisioned a 6G vision of ``intelligent connection of everything and digital twins'', and proposed the potential 6G application scenarios such as immersive extended reality, digital twins, global coverage, and ubiquitous intelligence \cite{Whitepaper1}.
These application scenarios have extremely high requirements for the reliability of wireless communication systems.
To achieve this goal, two methods are currently being discussed in academia, both of which aim to improve the reliability of communication system by actively reconfiguring channel parameters.
The first method is massive MIMO, which expands the dimensionality of the wireless channel by creating separate parallel point-to-point links, with the number of antennas expected to scale to hundreds or even thousands \cite{11026007}.
However, one of its key challenges is extremely high complexity and costly implementation, as each antenna typically requires a dedicated radio frequency (RF) feed.
The second method is programmable wireless environments (PWEs), which focus on enabling programmability of wireless channels and are expected to emerge as a new paradigm in 6G networks.
Specifically, PWEs employ precisely designed reconfigurable elements to directionally control the propagation of electromagnetic (EM) waves.

In the context of PWEs, one of the most representative technologies is reconfigurable intelligent surface (RIS), which is deployed between the transmitter and receiver and intelligently reshapes wireless channels by adjusting its phase shift \cite{9326394,10659326,10915642}.
In addition, the RIS capabilities can be expanded by improving its architecture. 
For example, simultaneous transmitting and reflecting RIS (STAR-RIS) can achieve 360-degree coverage \cite{9437234}, active RIS can amplify incident signals \cite{9998527}, holographic RIS can achieve holographic communication \cite{9374451}, zero-energy RIS can absorb the energy of incident waves to achieve self-sustainability \cite{10348506,11072046}, and multi-functional RIS can simultaneously achieve signal reflection, refraction, amplification, and energy harvesting functions \cite{10679239}.
Fluid antenna technique can flexibly switch positions within the antenna structure to locate the best channel, while movable antenna technique can dynamically alter antenna positions to improve channel conditions \cite{9264694,10286328}.
However, these reconfigurable technologies focus solely on enhancing channel gain and mitigating small-scale fading, neglecting path loss, which is a critical factor affecting the performance of wireless system.
Specifically, the double fading effect caused by RIS cascaded channels results in extremely high path loss during remote-distance transmission.
Additionally, fluid antennas and movable antennas can only adjust their positions within the wavelength scale range, and fail to address large-scale path loss.
It is worth noting that in the aforementioned technologies, the number of antennas is determined prior to system design, and is unable to be added or removed afterward.

Pinching antenna (PA) technique, as a groundbreaking paradigm to handle these limitations, converts path loss into a reconfigurable parameter to provide a new degree of freedom, while addressing the challenges of free-space path loss and line-of-sight (LoS) blockage \cite{11180028,11018390}.
The concept and prototype of pinching antenna system (PAS) were first proposed by NTT DOCOMO in 2021 \cite{NTTDOCOMO}.
Specifically, PAS uses dielectric waveguide as the transmission medium, and radiates EM waves from the waveguide into free-space through small separated dielectric elements termed as PA.
Due to the physical characteristic of PA, it can be dynamically deployed at any position on the waveguide to radiate or receive EM waves, and can be added or removed at any time according to user services and channel conditions to tailor the communication area and establish a strong LoS link.
Therefore, PA is expected to play an important role in 6G networks, and evolve into a new generation of transformative technology.


Recently, some scholars have conducted researches on PAS to explore its potential to meet the growing demands of 6G and beyond, but it is still in its infancy \cite{10896748,10909665,10945421,10912473,10976621,10981775,11106459,11096622,11206485}.
The authors in \cite{10896748} and \cite{10909665} optimized the downlink and uplink PAS transmission rate maximization problems respectively, emphasizing the enormous potential of PAS for robust and efficient communication in 6G networks.
In \cite{10945421}, the authors demonstrated the superior performance of non-orthogonal multiple access (NOMA) assisted PAS, and the authors in \cite{10912473} further investigated a NOMA assisted downlink PAS, aiming to find the activation location and number of pinching antennas (PAs) to maximize system throughput.
Moreover, the authors in \cite{10976621} analyzed the reliability of PAS and derived the optimal position of PA. 
In \cite{10981775}, the authors derived the upper bound of the PAS array gain under fixed antenna spacing, as well as the optimal number of antennas and optimal antenna spacing for maximizing the array gain.
Although the above studies provide a detailed analysis of information transmission (IT) in PAS, they neglect the importance of energy transfer (ET) in future communication systems. While 6G aims to provide seamless three-dimensional ubiquitous coverage and connectivity, it is also moving toward energy-sustainable and lightweight communication, which offers an infinite device lifetime without the need for manual battery charging or replacement from the user's perspective \cite{9193903}. In this direction, wireless-powered communication (WPC) eliminates the need for regular manual battery replacement, significantly extending the lifespan of wireless devices and paves the road for ture mobility \cite{9858871}.
The authors in \cite{11106459} considered a novel PAS assisted simultaneous wireless information and power transfer (SWIPT) to facilitate IT to multiple information receivers and ET to multiple energy receivers simultaneously, but they did not account for the impact of practical waveguide loss on PAS performance.
Different from SWIPT, wireless-powered communication (WPC) divides the entire communication process into two stages, namely the ET stage and the IT stage. The authors in \cite{11096622} introduced a wireless powered pinching antenna network, and proposed three approaches to simplify the resource allocation problem and then solved it efficiently using convex optimization methods.
Similarly, the authors in \cite{11206485} considered both time division multiple access (TDMA) and NOMA protocols in PA-aided WPC network.

However, in the context of PA, the existing researches have not considered the performance analysis of wireless-powered PAS.
Motivated by the consideration, this paper aims to develop a novel WPC system, utilizing two waveguides to mitigate the double near-far problem inherent in traditional WPC system, and offers a comprehensive evaluation of the performance improvements by PAS to achieve more efficient and reliable communications.
The key contributions of this paper are summarized as follows.

\begin{itemize}
    \item The characteristic of PA makes path loss a programmable parameter, which greatly mitigates the path loss of the communication link and emphasizes the high compatibility between PA and WPC.
    Accordingly, combining WPC and PA technologies, we introduce the concept of wireless-powered PAS, in which a user is randomly deployed within a rectangular area, to explore the advantages of PA in improving the reliability of WPC system.

    \item We consider the cases of practical lossy waveguide and ideal lossless waveguide, and derive the closed-form expressions for outage probability and ergodic rate with PA deployed in alignment with user, respectively, to evaluate the performance gap between wireless-powered PAS and traditional WPC system.
    Then, the optimal deployment analysis of waveguides is conducted to provide valuable insights for guiding their deployments.
    Moreover, the derived closed-form expressions closely aligns with the simulation results, validating the accuracy of the theoretical analysis of outage probability and ergodic rate.

    \item The theoretical and numerical results reveal the following conclusions:
    1) Compared with traditional WPC system, PA can be flexibly moved along the dielectric waveguide to mitigate the path loss and establish a strong LoS communication link;
    2) As the transmission power increases, the outage probability of wireless-powered PAS decreases sharply, while the ergodic rate of wireless-powered PAS increases steadily;
    3) The increases in the absorption coefficient and the dimensions of user area lead to an increase in the outage probability and a decrease in the ergodic rate of wireless-powered PAS;
    4) Under conditions of high absorption coefficient and long waveguide, the outage probability of wireless-powered PAS is even worse than that of traditional WPC system, while the ergodic rate of wireless-powered PAS is better than that of traditional WPC system;
    5) Interestingly, the wireless-powered PAS has an optimal distance between power station (PS) and access point (AP) as well as an optimal time allocation factor to minimize the outage probability or maximize the ergodic rate;
    6) The system performance of PS and AP separated at the optimal distance between PS and AP is superior to that of PS and AP integrated into a hybrid access point (HAP).
\end{itemize}

The remaining of this paper is organized as follows.
In Section \Rmnum{2}, we introduce the basic principles of PAS, and the concept of wireless-powered PAS.
Then, the channel and signal models for wireless-powered PAS are presented.
In Section \Rmnum{3}, we derive closed-form expressions for the outage probability and ergodic rate of wireless-powered PAS, including the practical lossy waveguide case and the ideal lossless waveguide case. Moreover, we conduct the position analysis for the waveguides and user.
In Section \Rmnum{4}, we present the simulation results and discussion to obtain the useful insights.
In Section \Rmnum{5}, we summarizes the findings and concludes the paper.

\section{System Model}
\subsection{System Setup}
As shown in Fig. \ref{system_model}, we consider the dual-link WPC scenario, i.e., downlink ET and uplink IT, and introduce a wireless-powered PAS, which consists of a PS, an AP, and an energy-constraint user (U) equipped with a single antenna.
Moreover, the PS and AP use PAs along the waveguides to establish stable LoS communication links and improve communication performance.
Different from the conventional electronic antennas, the phenomenon of electromagnetic (EM) coupling between the PA and the dielectric waveguide makes it possible to exchange power between them, and then EM waves are radiated into free-space.
Due to the passive nature of PA, its position can be flexibly adjusted on the waveguide for better channel conditions.
In this paper, it is assumed that the line connecting PS and AP is the $y$-axis (the direction of PS is positive), their symmetry axis is the $x$-axis, and the distance between PS and AP is $L$.
The waveguides of PS and AP are placed parallel to the $x$-axis at a height $h$, spanning a length equal to $D_x$.
Thus, the positions of the PS and AP waveguide feeding points can be denoted as ${\bm \psi}_p=\left( 0,\frac{L}{2},h \right)$ and ${\bm \psi}_a=\left( 0,-\frac{L}{2},h \right)$, respectively.
Both PS and AP have one PA on the waveguide, named PA-1 and PA-2, respectively. 
Since the PA can radiate signals at any position on the waveguide, the positions of PA-1 and PA-2 are denoted by ${\bm \psi}_{1}^{\rm pin}=\left( x_{1}^{\rm pin},\frac{L}{2},h \right)$ and ${\bm \psi}_{2}^{\rm pin}=\left( x_{2}^{\rm pin},-\frac{L}{2},h \right)$, where $x_{1}^{\rm pin},x_{2}^{\rm pin} \in \left[0,D_x\right]$.
In addition, we assume that U is randomly located in a rectangular region in the $x$-$y$ plane, whose lengths on the $x$-axis and $y$-axis are equal to $D_x$ and $D_y$, respectively.
The position of U is denoted by ${\bm \psi}_m = \left(x_m,y_m,0\right)$, where $x_m$ is uniformly distributed in $\left[0,D_x\right]$, and $y_m$ is uniformly distributed in $\left[-\frac{D_y}{2},\frac{D_y}{2}\right]$.

\begin{figure}[t]
    \centering
    \includegraphics[width=\linewidth]{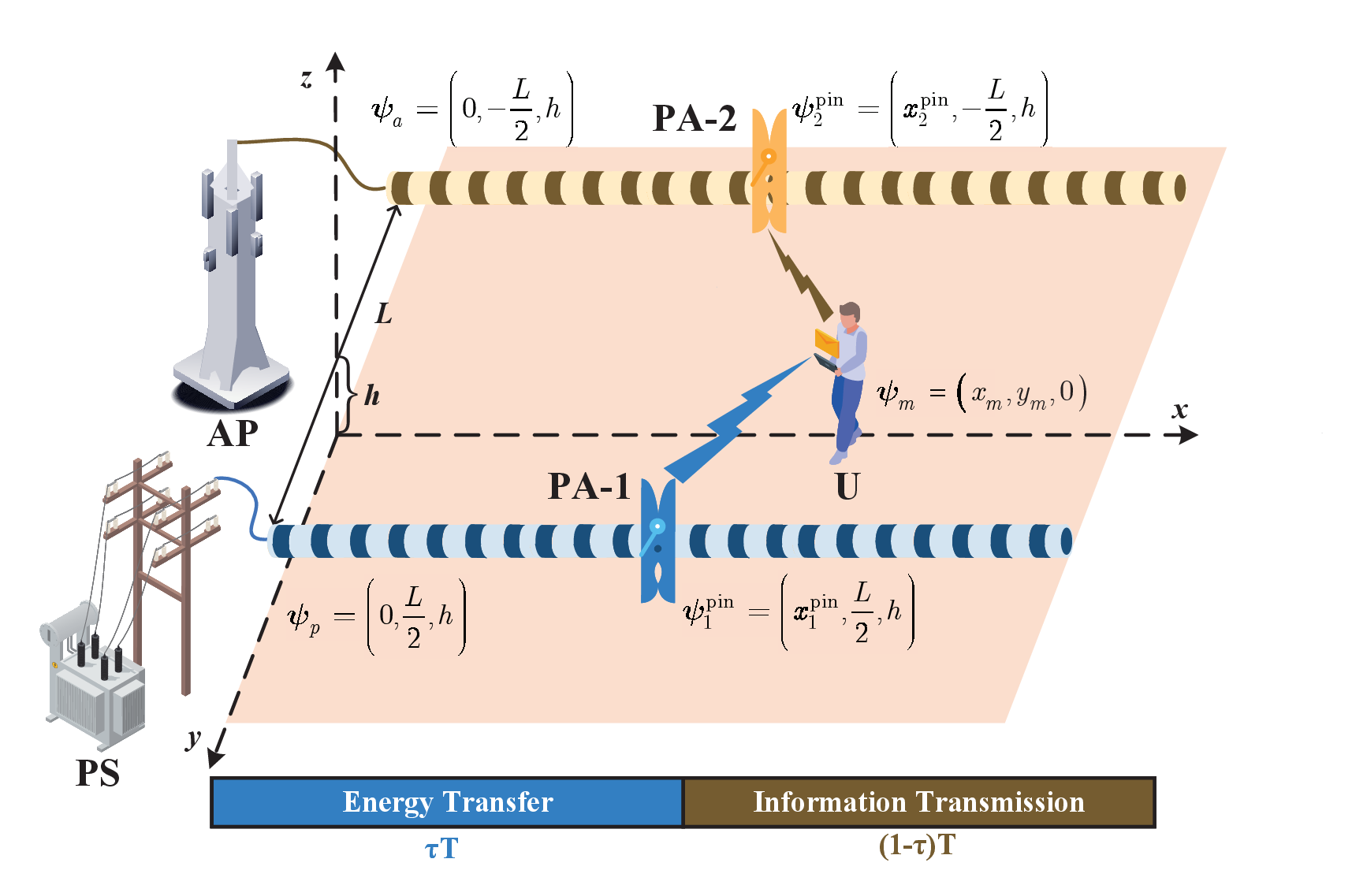}
    \caption{Wireless-powered pinching antenna system.}
    \label{system_model}
\end{figure}

\subsection{Channel Model}

According to the spherical-wave channel model, the free-space propagation channel between PA-$i$ and U is given by
\begin{align}
h_i=\frac{\sqrt{\beta} e^{-j\frac{2\pi}{\lambda} \left\lVert {\bm \psi}_m - {\bm \psi}_{i}^{\rm pin} \right\rVert} }{\left\lVert {\bm \psi}_m - {\bm \psi}_{i}^{\rm pin} \right\rVert },
\end{align}
where $i \in \{1,2\}$, $\beta=\left(\frac{\lambda}{4\pi}\right)^2$ denotes the path loss in free-space at a reference distance of 1 m, and $\left\lVert \cdot \right\rVert $ denotes the Euclidean norm.
Besides, $\lambda=\frac{c}{f_c}$ denotes the wavelength in free-space, $c$ denotes the speed of light, and $f_c$ denotes the carrier frequency.

Due to the propagation properties of dielectric material, when a signal propagates within a waveguide, the signal interacts with the dielectric core and the surrounding material within the waveguide, resulting in a change in the guided wavelength as $\lambda_g=\frac{\lambda}{n_{\rm eff}}$, where $n_{\rm eff}$ denotes the effective refractive index.
Therefore, the signal transmitted through the waveguide experiences an additional phase shift, which is given by
\begin{align}
g_i=e^{-j\frac{2\pi}{\lambda_g}\left\lVert {\bm \psi}_{i}^{\rm pin} - {\bm \psi}_j \right\rVert },
\end{align}
where $j=p$ when $i=1$, and $j=a$ when $i=2$.

Based on a rigorous derivation of Maxwell's equations, there is an exponential power attenuation of the signal as it passes through the waveguide, and the absorption coefficient $\alpha \in [0,+\infty)$ is used to characterize this phenomenon \cite{10976621}.
It is reasonable to assume that the absorption coefficient of the waveguides of PS and AP are the same.

\subsection{Signal Model}
Within a time frame $T$, the communication process is divided into two stages, i.e., ET stage and IT stage.
During the ET stage $\tau T \ (0<\tau<1)$, PS transfers energy directly to U.
By employing the linear EH model, the harvested energy of U can be expressed as
\begin{align}
    \label{Eu}
    &E_u = {\eta}E_t e^{-\alpha\left\lVert {\bm \psi}_{1}^{\rm pin} - {\bm \psi}_p \right\rVert } \left\lvert h_1 g_1 \right\rvert^{2} \nonumber
    \\
    &= \frac{ {\eta} \beta E_t e^{-\alpha\left\lVert {\bm \psi}_{1}^{\rm pin} - {\bm \psi}_p \right\rVert  }  \left\lvert e^{-j\left( \frac{2\pi}{\lambda} \left\lVert {\bm \psi}_m - {\bm \psi}_{1}^{\rm pin} \right\rVert + \frac{2\pi}{\lambda_g}\left\lVert {\bm \psi}_{1}^{\rm pin} - {\bm \psi}_p \right\rVert \right) } \right\rvert^2  }{ \left\lVert {\bm \psi}_m - {\bm \psi}_{1}^{\rm pin} \right\rVert^2 } \nonumber
    \\
    &\overset{(a)}{=} \frac{ {\eta} \beta E_t e^{-\alpha\left\lVert {\bm \psi}_{1}^{\rm pin} - {\bm \psi}_p \right\rVert  }   }{ \left\lVert {\bm \psi}_m - {\bm \psi}_{1}^{\rm pin} \right\rVert^2 },
\end{align}
where $E_t={\tau}{T}{P_s}$ denotes the total transmission energy of the PS in the ET phase, $\tau$ denotes time allocation factor between ET and IT stages, $P_s$ denotes the transmission power of the PS, $\eta$ denotes energy conversion efficiency, and the step $(a)$ uses the equation $\left\lvert e^{-jx}\right\rvert = 1$.

During the IT stage $(1-\tau)T$, U exploits the harvested energy to transmit its information to the AP.
The transmission power of U can be written as
\begin{align}
    \label{Pu}
    {P_u} = \frac{E_u}{(1-\tau )T}  =  \frac{ \beta {P_t} e^{-\alpha\left\lVert {\bm \psi}_{1}^{\rm pin} - {\bm \psi}_p \right\rVert  }   }{ \left\lVert {\bm \psi}_m - {\bm \psi}_{1}^{\rm pin} \right\rVert^2 },
\end{align}  
where ${P_t}=\frac{ \tau }{1-\tau } \eta P_s $. 

Then, the received signals at AP is given by
\begin{equation}
    \label{ya}
    {y_a}=\sqrt{P_u e^{-\alpha\left\lVert  {\bm \psi}_a - {\bm \psi}_{2}^{\rm pin}  \right\rVert  }} g_2 h_2 x_u + n_a,
\end{equation}
where $x_u$ is the transmitted signal of U, satisfying $\mathbb{E} [\lvert x_u\rvert^2] = 1 $, $\mathbb{E}\left[\cdot\right] $ denotes the expectation, and $n_a$ denotes the additive white Gaussian noise (AWGN) at AP.

Therefore, the signal-to-noise ratio (SNR) of signal received at AP can be expressed as
\begin{equation}
    \label{SNRa}
    {\rm{SNR}}_{a}= \frac{\beta^2 \rho_t e^{-\alpha \left(\left\lVert  {\bm \psi}_{1}^{\rm pin} - {\bm \psi}_p  \right\rVert + \left\lVert  {\bm \psi}_a - {\bm \psi}_{2}^{\rm pin}  \right\rVert \right) } }{\left\lVert {\bm \psi}_m - {\bm \psi}_{1}^{\rm pin} \right\rVert^2 \left\lVert {\bm \psi}_{2}^{\rm pin} - {\bm \psi}_m \right\rVert^2},
\end{equation}
where ${\rho }_t=\frac{P_t}{{\sigma}^2} $, and ${{\sigma}^2}$ is the variance of the AWGN.

\begin{remark}\label{remark1}
    When employing multiple pinching antennas for beamforming on each waveguide, these antennas can be deployed in close proximity to ensure constructive interference to the user. 
    Furthermore, given the extremely small order of magnitude of the wavelength $\lambda$, the distance between pinching antennas can be negligible, and their path loss relative to the user is virtually identical. 
    Assuming that the number of pinching antennas on the waveguides at the PS and AP are respectively $N_1$ and $N_2$, and the signal radiation for multiple pinching antennas follows the equal power model \cite{wang2025modelingbeamformingoptimizationpinchingantenna}.
    Therefore, according to Eq. (\ref{SNRa}), the SNR of signal received at AP with multiple pinching antennas can be closely approximated by using the expression for the single pinching antenna scenario, i.e., ${\rm{SNR}}_{a}' = N_1 N_2 {\rm{SNR}}_{a}$.
\end{remark}

\section{Performance Analysis}

In this section, we consider two cases, i.e., practical lossy waveguide case and ideal lossless waveguide case.
Further, we derive the closed-form expressions of outage probability and ergodic rate for the two cases, respectively, in order to comprehensively evaluate the communication performance of the wireless-powered PAS when the PAs are deployed at the location where their distance from the user is minimized, i.e., ${\bm \psi}_{1}^{\rm pin}=\left( x_m,\frac{L}{2},h \right)$, and ${\bm \psi}_{2}^{\rm pin}=\left( x_m,-\frac{L}{2},h \right)$.

\subsection{Lossy Waveguide Case}
In this subsection, we consider a practical case that in-waveguide propagation loss exists, i.e., $\alpha \neq 0$.

\subsubsection{Outage Probability}
Deploying the PA at $x_m$, i.e., aligning the PA with U, ensures minimal path loss, but there is a possibility of outage in communication due to the randomness of the position of U.
Therefore, the communication outage occurs when main channel capacity falls below the predefined target codeword rate, and then the outage probability is mathematically denoted as
\begin{equation}\label{OP}
    P_{out}={\rm Pr}[(1-\tau)\log_2(1+{\rm{SNR}}_{a})<R],
\end{equation}
where $R$ denotes the predefined target codeword rate.

\begin{proposition}\label{proposition1}
    The closed-form expressions for outage probability of wireless-powered PAS with lossy waveguide under different conditions are shown in Table \ref{table1}, which is shown at the top of next page, where $\chi=\frac{\beta^2 \rho_t}{\epsilon}$, $\epsilon=2^{\frac{R}{1-\tau} }-1$, $\varpi_k=\cos(\frac{2k-1}{2K}\pi )$, $\varphi_{k,1}=\frac{b}{2}(\varpi_k+1)$, $\varphi_{k,2}=\frac{b-c}{2}\varpi_k + \frac{b+c}{2}$, $\varphi_{k,3}=\frac{D_x}{2}(\varpi_k + 1)$, $\varphi_{k,4}=\frac{D_x-c}{2}\varpi_k + \frac{D_x+c}{2}$, $a = \frac{1}{2\alpha} \ln(\frac{\chi}{h^2 L^2})$, $b = \frac{1}{2\alpha} \ln \bigg( \frac{\chi}{  \left(h^2 - \frac{L^2}{4} \right)^2 + h^2 L^2} \bigg)$, $c = \frac{1}{2\alpha} \ln \bigg( \frac{\chi}{  \left( \frac{D_y^2}{4} + h^2 - \frac{L^2}{4} \right)^2 + h^2 L^2} \bigg)$, and $K$ is the accuracy versus complexity parameter.
\end{proposition}

\begin{table*}[t]
    \centering
    \captionsetup{font={normalsize}}
    \caption{The conditions and closed-form expressions for outage probability $P_{out}$ in lossy waveguide case}
    \label{table1}
    \resizebox{\linewidth}{!}{
    \begin{tabular}{l|l}
    \hline
    \hline
    \bf Condition & \bf The Closed-Form Expression
    \\
    \hline
    \hline
    $\left(h^2 - \frac{L^2}{4} \right)^2 + h^2L^2  \geqslant  \chi $  & 1
    \\
    \hline
    $\left(h^2 - \frac{L^2}{4} \right)^2 + h^2L^2  <  \chi$, $\left(h^2 - \frac{L^2}{4} \right)^2 + h^2L^2 >  \chi e^{-2 \alpha D_x}$, & $1 + \frac{\pi}{2\alpha D_x D_y K} \ln \left( \frac{\left( h^2 - \frac{L^2}{4} \right)^2 + h^2L^2}{\chi} \right) \sum_{k = 1}^{K} \sqrt{1-\varpi_k^2}$
    \\
    $\left(\frac{D_y^2}{4} + h^2 - \frac{L^2}{4} \right)^2 + h^2L^2 \geqslant  \chi$ and $ h^2 - \frac{L^2}{4} > 0$ & $\times \sqrt{ \sqrt{\chi e^{-2\alpha \varphi_{k,1}} - h^2 L^2} + \frac{L^2}{4} - h^2}$
    \\
    \hline
    $\left(h^2 - \frac{L^2}{4} \right)^2 + h^2L^2  <  \chi$, $\left(h^2 - \frac{L^2}{4} \right)^2 + h^2L^2 >  \chi e^{-2 \alpha D_x}$, & $ 1 + \frac{1}{2\alpha D_x}  \ln \left( \frac{\left(\frac{D_y^2}{4} + h^2 - \frac{L^2}{4} \right)^2 + h^2 L^2 }{\chi} \right)  $
    \\
    $\left(\frac{D_y^2}{4} + h^2 - \frac{L^2}{4} \right)^2 + h^2L^2 < \chi$, $ \left(\frac{D_y^2}{4} + h^2 - \frac{L^2}{4} \right)^2 + h^2L^2 > \chi e^{-2 \alpha D_x}$ & $-\frac{\pi}{2\alpha D_x D_y K} \ln \left( \frac{\left(\frac{D_y^2}{4} + h^2 - \frac{L^2}{4} \right)^2 + h^2L^2}{\left( h^2 - \frac{L^2}{4} \right)^2 + h^2L^2} \right) \sum_{k = 1}^{K} \sqrt{1-\varpi_k^2} $
    \\
    and $ h^2 - \frac{L^2}{4} > 0$ & $\times \sqrt{ \sqrt{\chi e^{-2\alpha \varphi_{k,2}} - h^2 L^2} + \frac{L^2}{4} - h^2}$
    \\
    \hline
    $\left(h^2 - \frac{L^2}{4} \right)^2 + h^2L^2 \leqslant \chi e^{-2 \alpha D_x}$, $\left(\frac{D_y^2}{4} + h^2 - \frac{L^2}{4} \right)^2 + h^2L^2 \geqslant  \chi$ & $1 - \frac{\pi}{D_y K} \sum_{k = 1}^{K} \sqrt{1-\varpi_k^2} \sqrt{ \sqrt{\chi e^{-2\alpha \varphi_{k,3}} - h^2 L^2} + \frac{L^2}{4} - h^2}$
    \\
    and $ h^2 - \frac{L^2}{4} > 0$ &
    \\
    \hline
    $\left(h^2 - \frac{L^2}{4} \right)^2 + h^2L^2 \leqslant \chi e^{-2 \alpha D_x}$, $\left(\frac{D_y^2}{4} + h^2 - \frac{L^2}{4} \right)^2 + h^2L^2 < \chi$, & $ \left( 1 + \frac{1}{2\alpha D_x}  \ln \left( \frac{\left(\frac{D_y^2}{4} + h^2 - \frac{L^2}{4} \right)^2 + h^2 L^2 }{\chi} \right)  \right) \bigg( 1 - \frac{\pi}{D_y K} $
    \\
    $ \left(\frac{D_y^2}{4} + h^2 - \frac{L^2}{4} \right)^2 + h^2L^2 > \chi e^{-2 \alpha D_x}$ and $ h^2 - \frac{L^2}{4} > 0$ & $\times \sum_{k = 1}^{K} \sqrt{1-\varpi_k^2} \sqrt{ \sqrt{\chi e^{-2\alpha \varphi_{k,4}} - h^2 L^2} + \frac{L^2}{4} - h^2} \bigg)$
    \\
    \hline
    $ \left(\frac{D_y^2}{4} + h^2 - \frac{L^2}{4} \right)^2 + h^2L^2 \leqslant \chi e^{-2 \alpha D_x}$ & 0
    \\
    \hline
    \hline
    \end{tabular}
    }
\end{table*}

\begin{proof}
    See Appendix \ref{AppendixA}.
\end{proof}

\subsubsection{Ergodic Rate}
Different the outage probability, the ergodic rate reflects the average achievable rate of the system, thus it can be used to quantify the information transmission capability of wireless-powered PAS.
In the following proposition, we derive the closed-form expression for ergodic rate of wireless-powered PAS to evaluate the potential of wireless-powered PAS in providing high transmission rates.
The ergodic rate is mathematically denoted as
\begin{equation}\label{Avg}
    R_p=\mathbb{E} \left[ (1-\tau) \log_2 \left( 1 + {\rm{SNR}}_{a} \right) \right],
\end{equation}
where $\mathbb{E} \left[ \cdot \right]$ denotes expectation.

\begin{proposition}\label{proposition2}
    The closed-form expression for ergodic rate of wireless-powered PAS with lossy waveguide is given by
    \begin{align}\label{Avg_result}
        R_p =& \frac{\pi(1-\tau)}{4 K D_x \alpha \ln2}\sum_{k = 1}^{K} \sqrt{1-\varpi_k^2} \bigg[ {\rm Li}_2\left( -\frac{\beta^2 \rho_t e^{-2\alpha D_x}}{p(u_k)} \right) \nonumber
        \\
        &- {\rm Li}_2\left( -\frac{\beta^2 \rho_t}{p(u_k)} \right) \bigg],
    \end{align}
    where ${\rm Li}_2\left( x \right) = -\int_{0}^{x} \frac{1-u}{u} \,du $ is the dilogarithm function \cite[Eq. (6.254.1)]{10.1115/1.3138251}, $p(u_k) = \left(u_k^2 + h^2 - \frac{L^2}{4} \right)^2 + h^2L^2$, and $u_k=\frac{D_y}{4}(\varpi_k + 1)$.
\end{proposition}

\begin{proof}
    See Appendix \ref{AppendixB}.
\end{proof}

\begin{remark}\label{remark2}
    The reason for the complexity of the closed-form expressions for both the outage probability and the ergodic rate of the wireless-powered PAS is due to the double transmission links and the randomness of U's location.
    It is observed from Proposition \ref{proposition1} and Proposition \ref{proposition2} that the outage probability is a monotonically increasing function with respect to $D_x$, $D_y$ and $\alpha$, whereas the ergodic rate is a monotonically decreasing function with respect to $D_x$, $D_y$ and $\alpha$.
    In other words, as $D_x$, $D_y$ and $\alpha$ increase, the outage probability increases while the ergodic rate decreases.
\end{remark}

\subsection{Lossless Waveguide Case}
In this subsection, we consider an ideal case that in-waveguide propagation loss is negligible, i.e., $\alpha = 0$, to obtain a lower bound on the outage probability and an upper bound on the ergodic rate of the wireless-powered PAS.

\subsubsection{Outage Probability}
\begin{proposition}\label{proposition3}
    The closed-form expressions for outage probability of wireless-powered PAS with lossless waveguide under different conditions are shown in Table \ref{table2}.
\end{proposition}

\begin{table*}[t]
    \centering
    \captionsetup{font={normalsize}}
    \caption{The conditions and closed-form expressions for outage probability $P_{out}$ in lossless waveguide case}
    \label{table2}
    \begin{tabular}{l|l}
    \hline
    \hline
    \bf Condition & \bf The Closed-Form Expression
    \\
    \hline
    \hline
    $  \left(h^2 - \frac{L^2}{4} \right)^2 + h^2L^2  >  \chi $  & 1
    \\
    \hline
    $ h^2 - \frac{L^2}{4} > 0$, $\left(h^2 - \frac{L^2}{4} \right)^2 + h^2L^2  \leqslant   \chi $ and $\left(\frac{D_y^2}{4} + h^2 - \frac{L^2}{4} \right)^2 + h^2L^2 \geqslant  \chi$ & $1 - \frac{2}{D_y} \sqrt{ \sqrt{\chi - h^2 L^2} + \frac{L^2}{4} - h^2} $
    \\
    \hline
    $ \left(\frac{D_y^2}{4} + h^2 - \frac{L^2}{4} \right)^2 + h^2L^2 < \chi $ & 0
    \\
    \hline
    \hline
    \end{tabular}
\end{table*}

\begin{proof}
    See Appendix \ref{AppendixC}.
\end{proof}

\subsubsection{Ergodic Rate}

\begin{proposition}\label{proposition4}
    The closed-form expression for ergodic rate of wireless-powered PAS with lossless waveguide is given by
    \begin{align}\label{Avg_result_lossless}
        R_p =& \frac{2(1-\tau)}{D_y \ln2} \Big( r_1 \ln \left( r_1^2 + 4s^2 \right) + r_2 \ln \left( r_2^2 + 4s^2 \right) \nonumber
        \\
        &- r_3 \ln \left( r_3^2 + 4h^2 \right) - r_4 \ln \left( r_4^2 + 4h^2 \right) \nonumber
        \\
        &+ 4s \left( \arctan \left(  \frac{r_1}{2s} \right) + \arctan \left(  \frac{r_2}{2s} \right) \right) \nonumber
        \\
        &- 4h \left( \arctan \left(  \frac{r_3}{2h} \right) + \arctan \left(  \frac{r_4}{2h} \right) \right) \Big),
    \end{align}
    where $r_1 = D_y + v$, $r_2 = D_y - v$, $r_3 = D_y + L$, $r_4 = D_y - L$, $v=\sqrt{2\left( u - h^2 + \frac{L^2}{4} \right)}$, $u = \sqrt{ \left( h^2 - \frac{L^2}{4} \right)^2 + h^2L^2 + \beta^2 \rho_t }$, and $s = \sqrt{u - \frac{v^2}{4}}$.
\end{proposition}

\begin{proof}
    See Appendix \ref{AppendixD}.
\end{proof}

\begin{remark}\label{remark3}
    It can be concluded from Proposition \ref{proposition3} and Proposition \ref{proposition4} that the outage probability and ergodic rate of the wireless-powered PAS in lossless waveguide case are independent of the coordinate of U on the $x$-axis $x_m$.
    This is because the PA is deployed at $x_m$ and $\alpha = 0$, which leads to the absence of a variable $x_m$ in the SNR of signal received at AP.
    Therefore, the conditions and closed-form expressions for the outage probability and ergodic rate become concise.
\end{remark}

\begin{remark}\label{remark4}
    For the multiple pinching antenna scenario, simply replace $\rho_t$ with $\rho_t N_1 N_2$ yields the closed-form expressions for outage probability and ergodic rate under both practical lossy waveguide and ideal lossless waveguide cases.
\end{remark}

\section{Position Analysis of Waveguides and User}
\subsection{Optimal Deployment Analysis of Waveguides}
In the previous analysis, we assume that the distance between waveguides of PS and AP $L$ is constant to facilitate the derivation of the closed-form expressions.
However, in the practical deployment, as shown in Eqs. (\ref{SNRa}) and (\ref{OP2}), the SNR of signal received at AP exhibits a non-linear change with respect to $L$.
Therefore, there is an optimal distance $L^{*}$ to maximize the SNR of signal received at AP.

\begin{proposition}\label{proposition5}
    The SNR of signal received at AP in wireless-powered PAS is maximized when the optimal distance between the two waveguides is given by
        \begin{equation}\label{Lopt}
        L^{*} = \left\{
            \begin{aligned}
            &0,\quad y_m \in  (-h,h) \\
            &2\sqrt{y_m^2 - h^2}, \quad y_m \in \left[-\frac{D_y}{2},-h\right] \cup  \left[h,\frac{D_y}{2}\right] \\
            \end{aligned}
        \right.
    \end{equation}
\end{proposition}

\begin{proof}
    As shown in Eqs. (\ref{SNRa}) and (\ref{OP2}), we have
    \begin{align}\label{SNRa2}
        {\rm{SNR}}_{a} = \frac{\beta^2 \rho_t e^{-2\alpha x_m}}{ \left( y_m^2 + h^2 - \frac{L^2}{4} \right)^2 + h^2 L^2 }.
    \end{align}

    It can be seen from (\ref{SNRa2}) that only the denominator in ${\rm{SNR}}_{a}$ contains the terms related to $L$.
    Let $f(L)=\left( y_m^2 + h^2 - \frac{L^2}{4} \right)^2 + h^2 L^2$, and then we can obtain the maximum ${\rm{SNR}}_{a}$ by minimizing $f(L)$.
    Therefore, the original problem is transformed into finding the value of $L$ corresponding to the minimum value of $f(L)$.

    We note that the first order derivative function of $f(L)$ is given by
    \begin{align}\label{fL_first_order}
        f'(L) = \left ( \frac{L^2}{4} + h^2 - y_m^2  \right ) L.
    \end{align}

    Since $L$ represents the distance between two waveguides, it can not be negative, i.e., $L \geqslant 0$.
    When $h^2 - y_m^2 > 0$, $f'(L) \geqslant 0$.
    Hence, $f(L)$ is a monotonically increasing function with respect to $L$, and the minimum value of $f(L)$ is obtained when $L = 0$, i.e., $L^{*} = 0$.
    In addition, when $h^2 - y_m^2 \leqslant 0$, it can be observed from (\ref{fL_first_order}) that when $L$ is small and close to $0$, $f'(L)<0$, and when $L$ is large and close to $D_y$, $f'(L)>0$.
    Thus, we have $f'(L^{*}) =0$, and $L^{*} = 2\sqrt{y_m^2 - h^2}$.
    The proof of Proposition \ref{proposition5} is completed.
\end{proof}

\begin{remark}\label{remark5}
    It is observed from Eq. (\ref{Lopt}) that when the user's position is determined, the SNR of signal received at AP can be maximized by adjusting the distance between waveguides of PS and AP.
    In particular, for the user located near the $x$-axis, the optimal distance between waveguides is equal to 0, i.e., PS and AP are integrated into a HAP.
    The upper limit of the distance between waveguides of PS and AP is $\sqrt{D_y^2 - 4h^2}$, indicating that a larger distance on $L$ is meaningless and only degrades the received signal of AP.
    The above conclusion depends on the joint effect of ET and IT losses.
    Intuitively, the optimal deployment position for the waveguide is at $L=0$, i.e., HAP.
    When the user is close to the $x$-axis, both ET and IT losses are relatively small, and the waveguide does not alter its optimal position.
    However, when the user moves away from the $x$-axis, the total ET and IT losses generated by deploying the waveguides for PS and AP separately are smaller than those for HAP.
\end{remark}

\subsection{Position Analysis of User}
Similar to the optimal deployment analysis of waveguides, the instantaneous SNR of signal received at AP also exhibits a non-linear variation about the user's position.
By fixing the distance between waveguides of PS and AP $L$, we can reveal the optimal position of the user for the SNR at AP.

\begin{figure}[t]
    \centering
    \includegraphics[width=\linewidth]{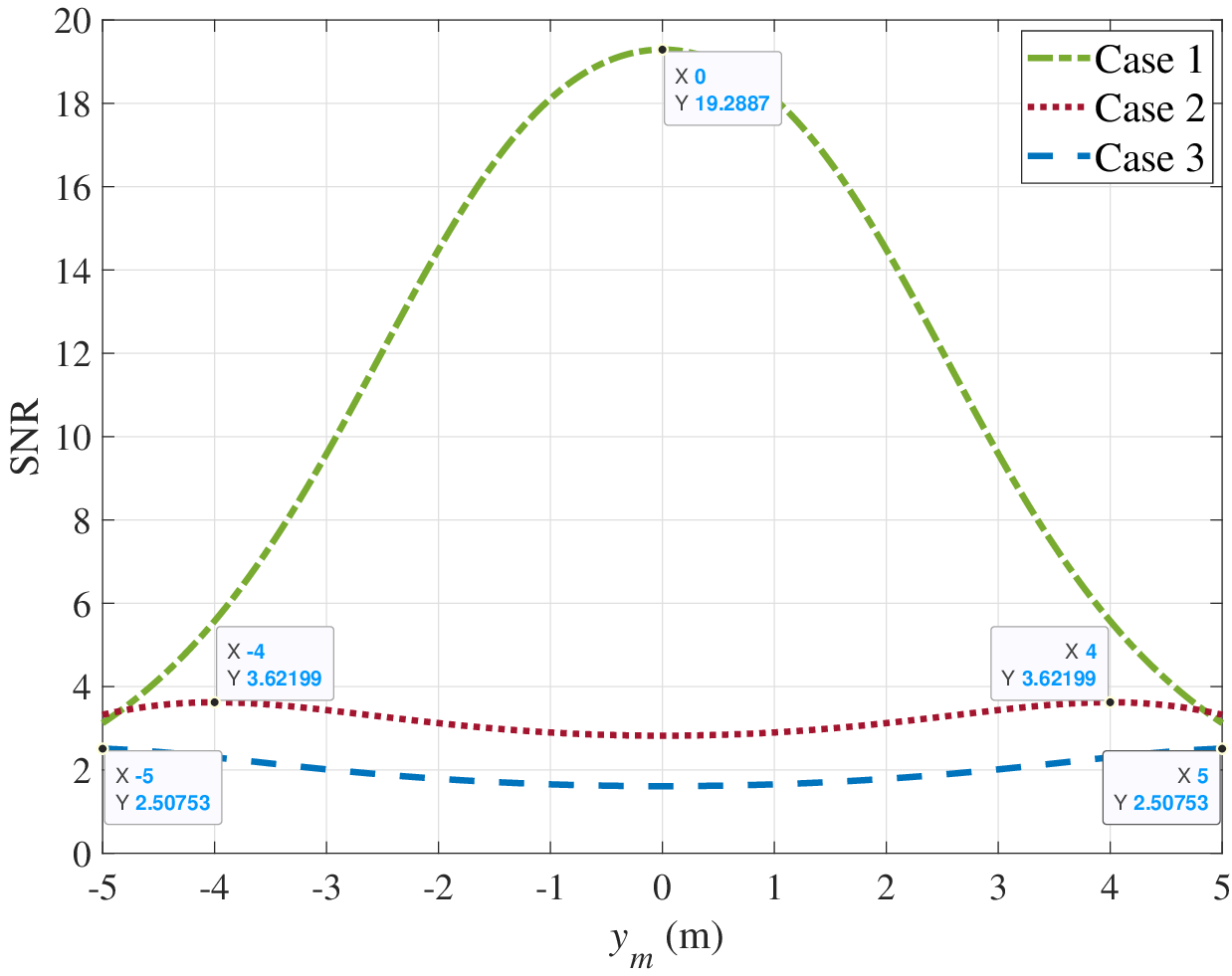}
    \caption{The SNR of signal received at AP versus $y_m$ for different cases, where $D_y=10$ m, $h=3$ m and $L=4$ m in Case 1, $h=3$ m and $L=10$ m in Case 2, and $h=3$ m and $L=12$ m in Case 3.}
    \label{SNR_ym}
\end{figure}

\begin{proposition}\label{proposition6}
    The SNR of signal received at AP in wireless-powered PAS is maximized when the position of user is given by
        \begin{align}\label{yopt}
            (x_m^{*},y_m^{*} )= \left\{
            \begin{aligned}
            &(0,0), &L \in  (0,2h)\\
            &\left(0,\pm \sqrt{\frac{L^2}{4} - h^2}\right), &L \in  \left[2h,\sqrt{D_y^2 + 4h^2} \right] \\
            &\left(0,\pm  \frac{D_y}{2}\right), &L \in  \left(\sqrt{D_y^2 + 4h^2},\infty \right) \\
            \end{aligned}
        \right.
    \end{align}
\end{proposition}
\begin{proof}
    According to Eq. (\ref{SNRa2}), it is readily apparent that the SNR of signal received at AP in wireless-powered PAS is maximized when $x_m=0$, i.e., $x_m^{*}=0$, and only the denominator contains the terms related to $y_m$.
    Similar to the proof of Proposition \ref{proposition5}, we can maximize the ${\rm{SNR}}_{a}$ by minimizing $f(y_m)=\left( y_m^2 + h^2 - \frac{L^2}{4} \right)^2 + h^2 L^2$.
    
    Taking the derivative with respect to $f(y_m)$, we have
    \begin{align}\label{fym_first_order}
        f'(y_m) = 4 y_m \left( y_m^2 + h^2 - \frac{L^2}{4} \right).
    \end{align}

    \textbf{Case 1:} When $h^2 - \frac{L^2}{4} > 0$, $f'(y_m)$ shares the same sign relationship with $y_m$.
    Since $y_m \in \left[-\frac{D_y}{2},\frac{D_y}{2}\right]$, $f(y_m)$ first decreases and then increases.
    Thus, $f(y_m)$ has a minimum value at $y_m=0$, i.e., $y_m^{*}=0$.
    
    \textbf{Case 2:} When $h^2 - \frac{L^2}{4} \leqslant  0$ and $\frac{D_y}{2} \geqslant \sqrt{\frac{L^2}{4} - h^2}$, we set $f'(y_m)=0$ to obtain the solutions $y_{1} = -\sqrt{\frac{L^2}{4} - h^2}$ and $y_{2} = \sqrt{\frac{L^2}{4} - h^2}$.
    When $y_m \in \left(-\frac{D_y}{2},y_1 \right) \cup \left(0,y_2 \right)$, $f(y_m)$ is a monotonically decreasing function, whereas when $y_m \in \left(y_1,0 \right) \cup \left(y_2,\frac{D_y}{2} \right)$, $f(y_m)$ is a monotonically increasing function.
    Therefore, we find that $f(y_1) = f(y_2) = h^2L^2$, which are the minimum value of $f(y_m)$, i.e., $y_m^{*}=\pm \sqrt{\frac{L^2}{4} - h^2}$.

    \textbf{Case 3:} When $h^2 - \frac{L^2}{4} \leqslant  0$ and $\frac{D_y}{2} < \sqrt{\frac{L^2}{4} - h^2}$, $f(y_m)$ first increases and then decreases, the minimum value of $f(y_m)$ is $f(-\frac{D_y}{2}) = f(\frac{D_y}{2})=\left( \frac{D_y^2}{4} + h^2 - \frac{L^2}{4} \right)^2 + h^2 L^2$, i.e., $y_m^{*}=\pm \frac{D_y}{2}$.
    The proof of Proposition \ref{proposition6} is completed.
\end{proof}

\begin{remark}\label{remark6}
It can be concluded from Eq. (\ref{yopt}) that the optimal positions for user occur in pairs and are symmetrical about the $x$-axis, which can be verified in Fig. \ref{SNR_ym}.
This is because the PS and AP waveguides are deployed at the same height and at equal distances from the $x$-axis, and the power attenuation effect resulting from path loss in both energy transfer and information transmission is identical.

\end{remark}

\section{Simulation Results and Discussion}
In this section, we evaluate the communication performance of the wireless-powered PAS, and employ Monte Carlo simulations with $10^6$ realizations to validate the theoretical analysis.
To validate the performance improvement of the PA in terms of system reliability, a traditional WPC system without the PA is considered as benchmark.
Unless otherwise stated, the carrier frequency $f_c = 2.7$ GHz, the effective refractive index $n_{\rm eff} = 1.4$, the absorption coefficient $\alpha = 0.01$, and the dimensions of user area $D_x = D_y = 10$ m.
Without loss of generality, the simulation parameters are set as $T=1$, $\tau=0.4$, $\eta =0.8$, $R=2.5$ BPCU, $h=3$ m, $L=4$ m, $\sigma ^2=-90$ dBm, and $K=50$.

\begin{figure}[t]
    \centering
    \includegraphics[width=\linewidth]{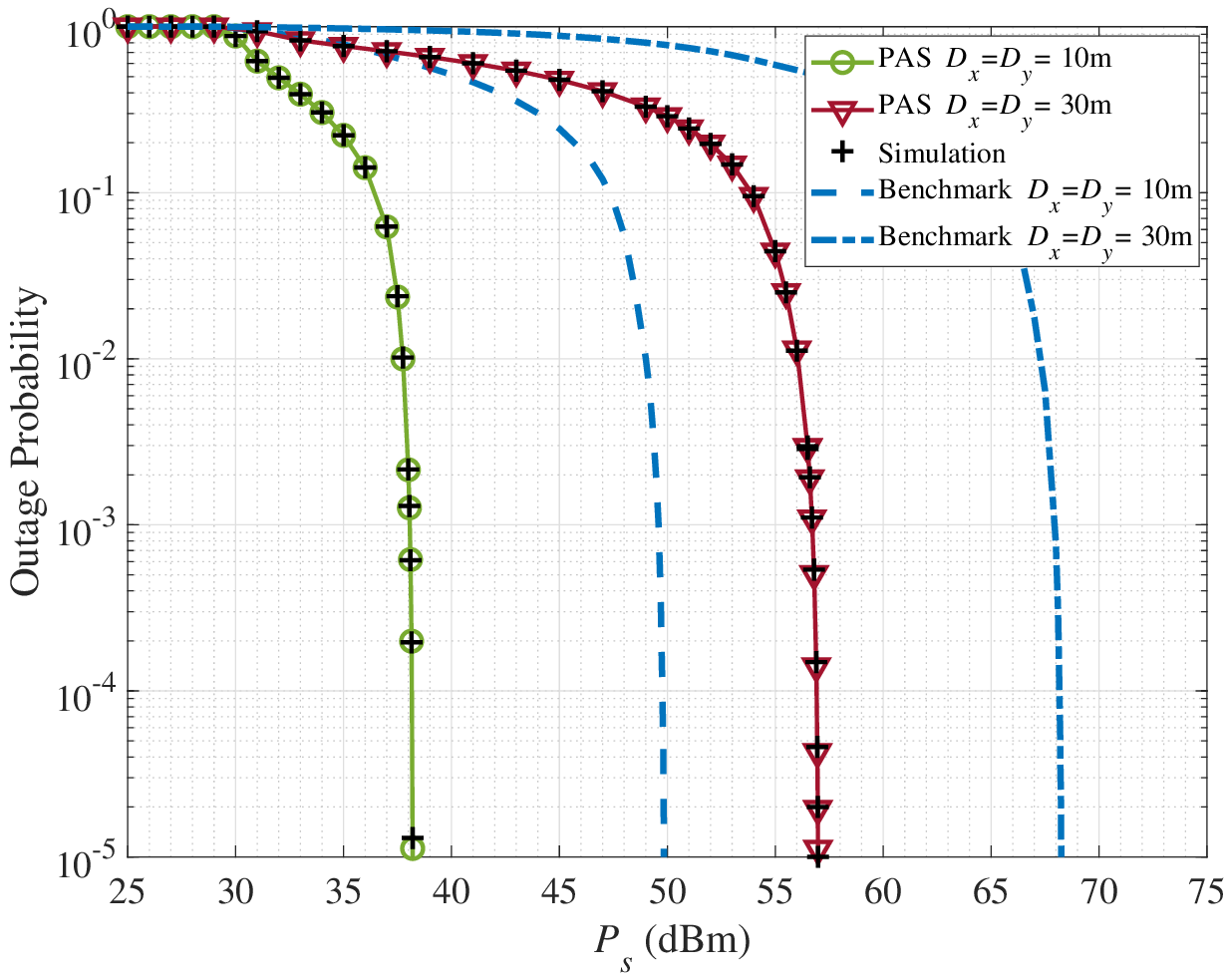}
    \caption{Outage probability versus $P_s$ for different $D_x$ and $D_y$.}
    \label{OP_Ps_DxDy}
\end{figure}

Fig. \ref{OP_Ps_DxDy} shows that the outage probability versus the transmission power $P_s$ for different dimensions of user area $D_x$ and $D_y$.
It is observed from the figure that the closed-form expressions of outage probability match well with the simulation results, verifying the correctness of the theoretical analysis of outage probability.
As $P_s$ increases, the outage probability decreases significantly.
In addition, under the same dimensions of user area, the wireless-powered PAS demonstrates an evident performance advantage over the traditional WPC system.
This is attributed to the fact that the PA can be flexibly moved along the dielectric waveguide to adjust its position to minimize the distance from U, effectively reducing path loss and establishing a strong LoS communication link.
However, for achieving an outage probability of $10^{-5}$, the transmission power gain required for the wireless-powered PAS is almost the same as that of the traditional WPC system when the user area dimensions are increased from $10$ m to $30$ m, which is different from the results of Fig. 2 in \cite{10976621}.
The reason for this conclusion is that the waveguides of the PS and AP are arranged on both sides of the user area, and the unique double transmission links in WPC increase the propagation loss within the waveguide and in the free space, which diminishes the performance advantage provided by the PA.

\begin{figure}[t]
    
    \centering
    \subfloat[$D_x=D_y=10$ m]{\includegraphics[width=\linewidth]{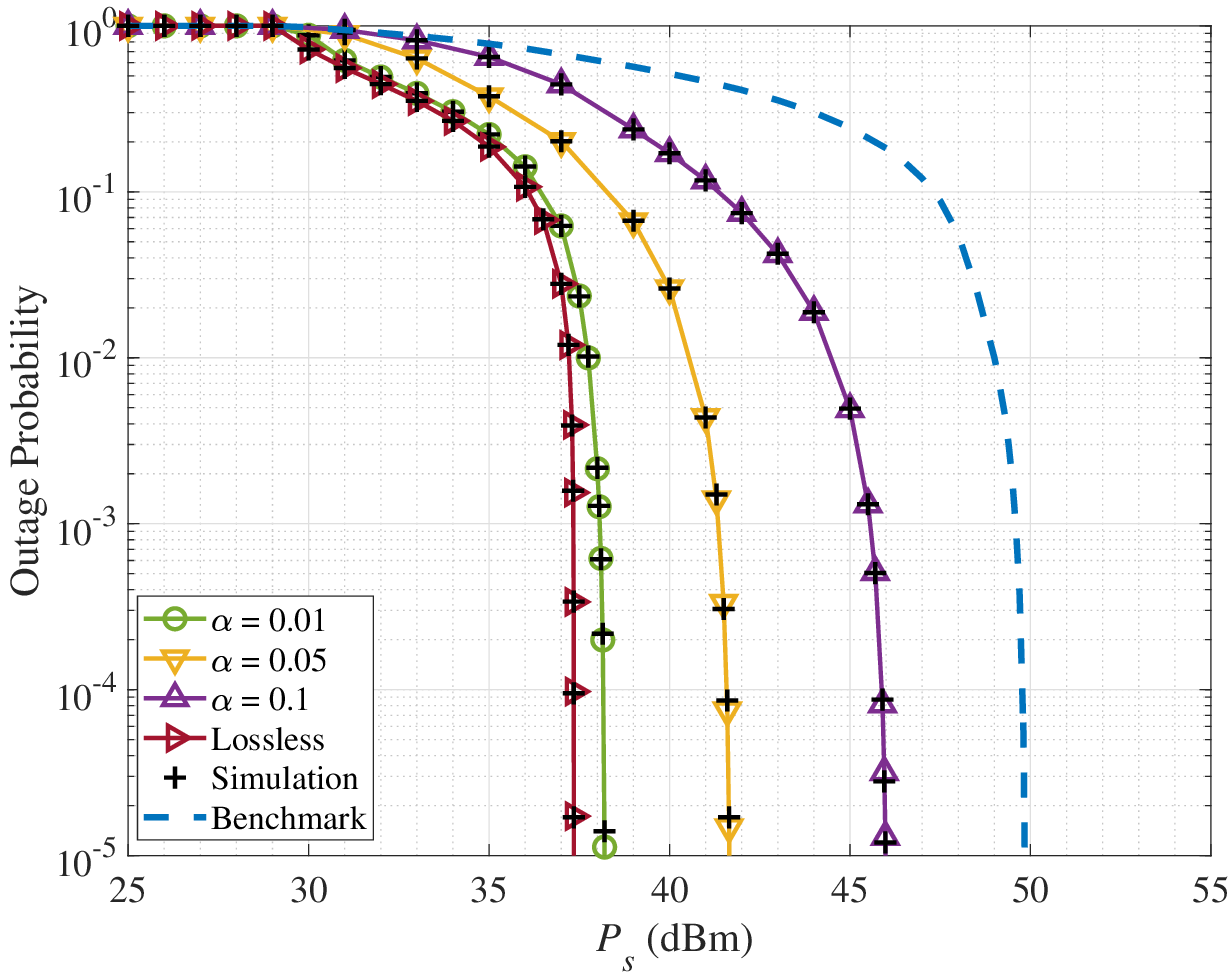}}\\
    \subfloat[$D_x=D_y=30$ m]{\includegraphics[width=\linewidth]{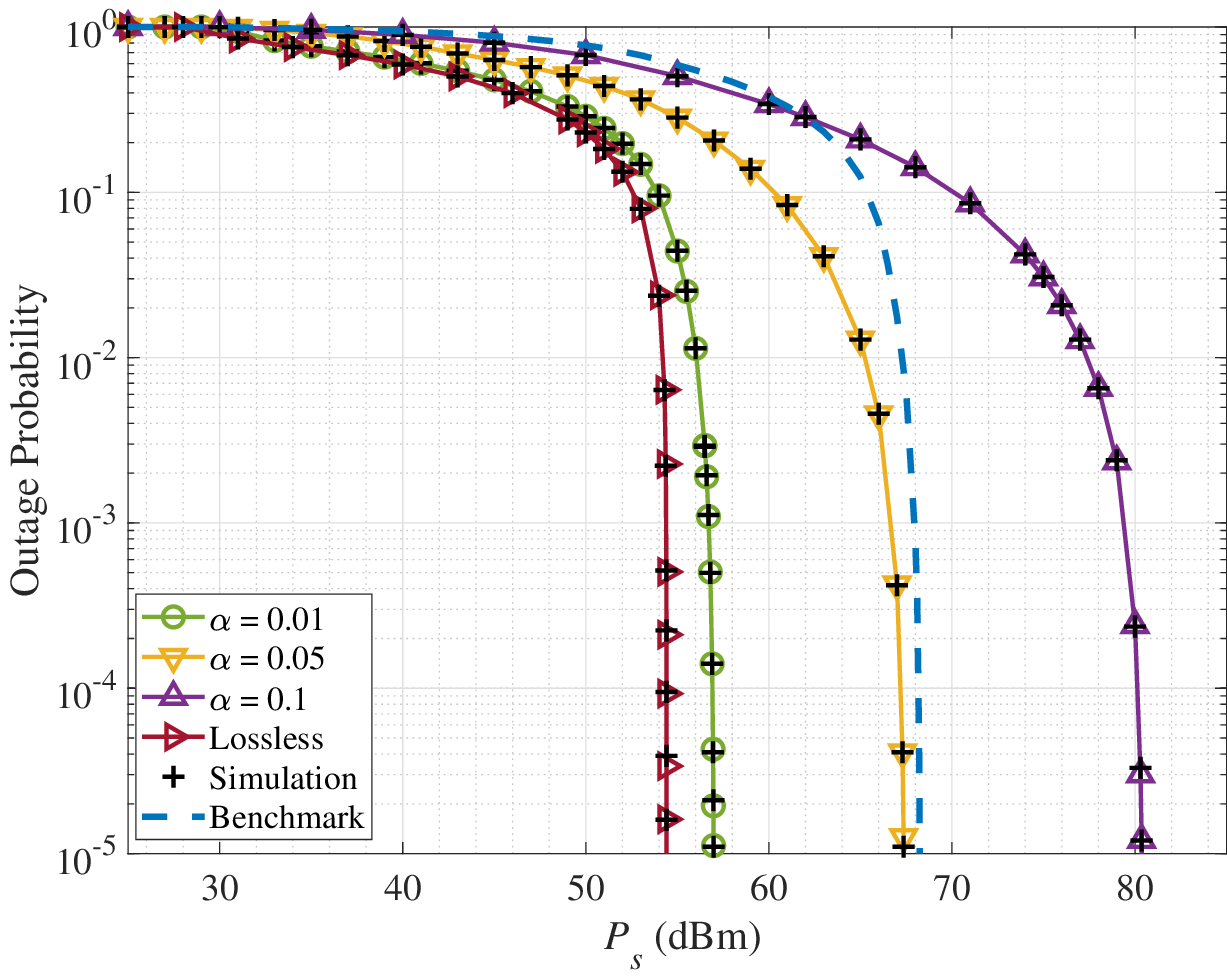}}
    
    \caption{Outage probability versus $P_s$ for different $\alpha$, where (a) $D_x=D_y=10$ m, and (b) $D_x=D_y=30$ m.}
    \label{OP_Ps_alpha}
\end{figure}

Fig. \ref{OP_Ps_alpha} illustrates that the outage probability versus the transmission power $P_s$ for different absorption coefficient $\alpha$, including the ideal case of a lossless waveguide, i.e., $\alpha=0$.
Fig. \ref{OP_Ps_alpha}(a) and Fig. \ref{OP_Ps_alpha}(b) represent the cases of $D_x = D_y = 10$ m and $D_x = D_y = 30$ m, respectively, where $\alpha$ varies in the range of $0.01$ to $0.1$.
It can be seen from Fig. \ref{OP_Ps_alpha}(a) that the outage probability of all cases is lower than the benchmark, and the lossless waveguide case achieves the best outage probability, demonstrating the importance of the PA in improving communication performance even with moderately lossy waveguide.
Moreover, the propagation loss in the waveguide increases with $\alpha$, which leads to an increase in outage probability.
Unexpectedly, Fig. \ref{OP_Ps_alpha}(b) shows that the outage probability increases for all cases as the dimensions of user area $D_x$ and $D_y$ increase, especially for the case $\alpha = 0.1$, whose outage probability is even higher than that of the benchmark.
This is due to the fact that under high $\alpha$ and $D_x$, the propagation loss in the waveguide is higher than the path loss in free-space, indicating the necessity of focusing on the in-waveguide propagation loss and the waveguide length in the practical PAS design.
Generally, Figs. \ref{OP_Ps_alpha}(a) and \ref{OP_Ps_alpha}(b) emphasize the characteristic of the PA to provide reliable communications, but this capability is weakened by high waveguide loss and a broader user area.

\begin{figure}[t]
    \centering
    \includegraphics[width=\linewidth]{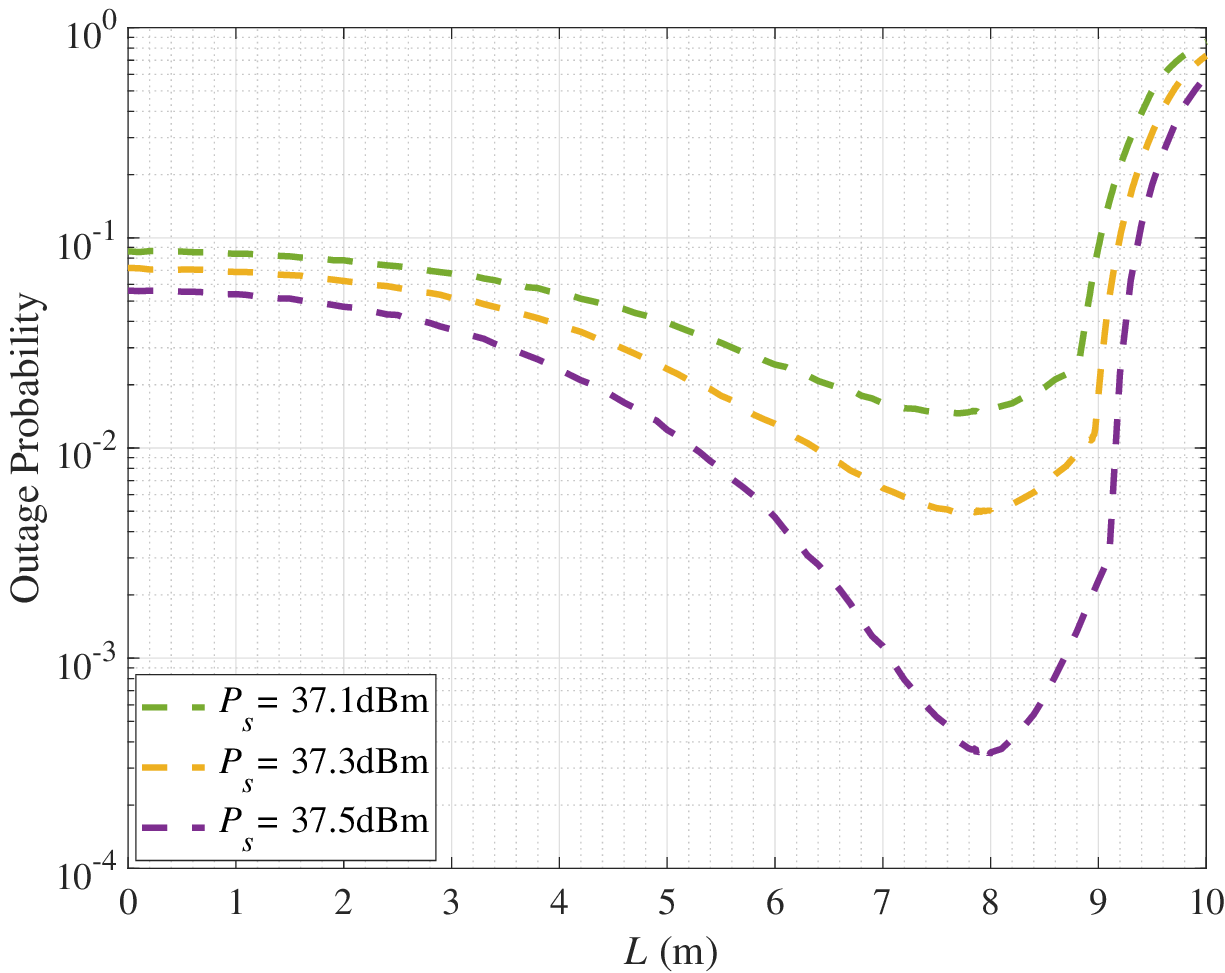}
    \caption{Outage probability versus $L$ for different $P_s$, where $D_x=D_y=10$ m.}
    \label{OP_L_diff_Ps}
\end{figure}

\begin{figure}[t]
    \centering
    \includegraphics[width=\linewidth]{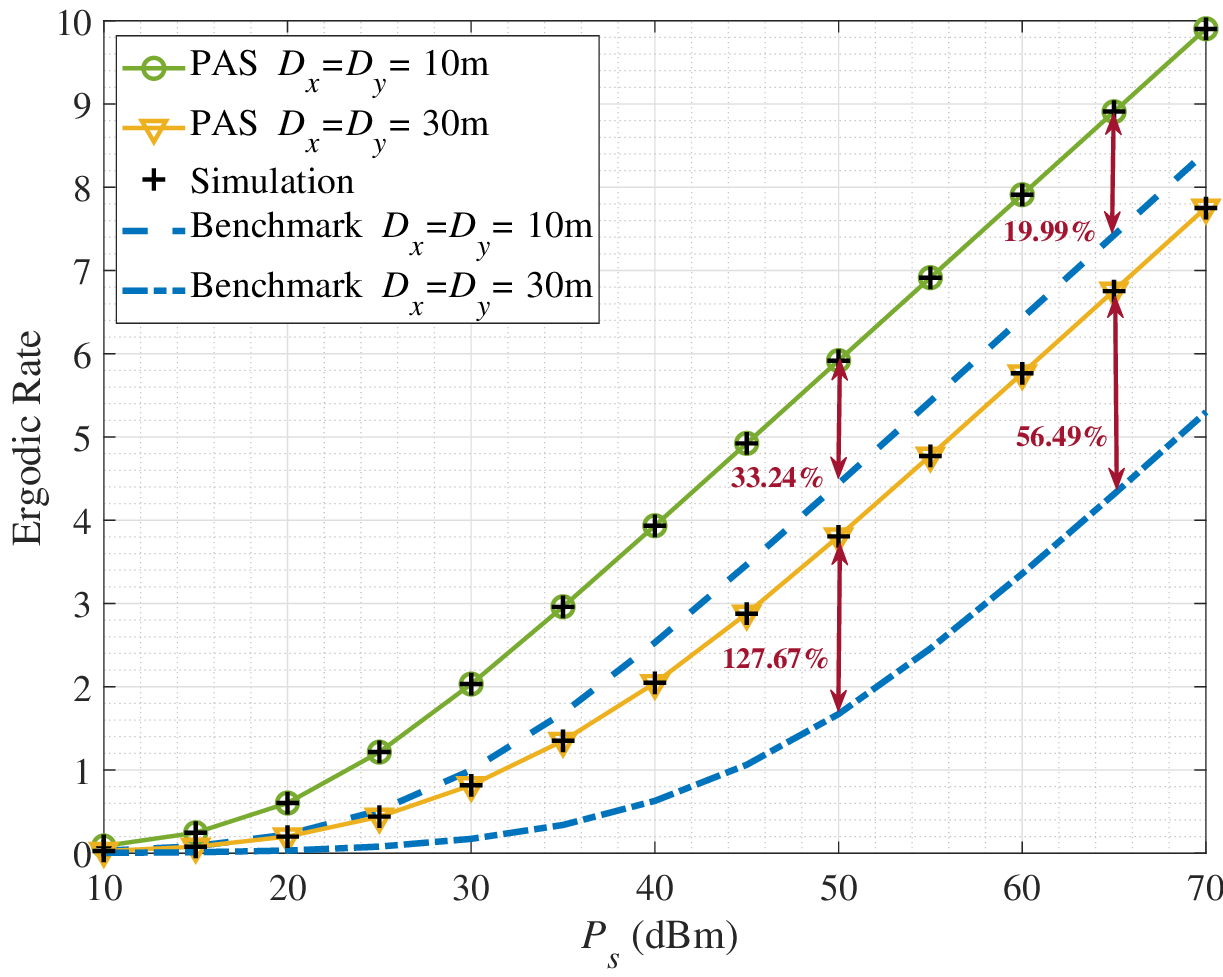}
    \caption{Ergodic rate versus $P_s$ for different $D_x$ and $D_y$.}
    \label{Avg_rate_Ps_DxDy}
\end{figure}

Fig. \ref{OP_L_diff_Ps} shows that the outage probability versus the distance between PS and AP $L$ for different transmission power $P_s$.
From this figure, we see that the outage probability of the wireless-powered PAS first decreases and then increases as $L$ increases, revealing that the wireless-powered PAS has an optimal distance between PS and AP $L^{opt}$ to minimize the outage probability.
The outage probability at $L^{opt}$ is much lower than that at $L=0$, which confirms that the system performance of PS and AP separated at an appropriate distance is superior to that of PS and AP integrated into a HAP.
In particular, when $L$ is greater than $L^{opt}$, the reliability of wireless-powered PAS deteriorates rapidly, and when $L$ increases to approach $D_y$, its performance is even worse than that of the HAP.
Consequently, it is essential to correctly design the distance between the PS and AP in wireless-powered PAS, as it significantly affects the PA's capability to achieve a reliable communication.

\begin{figure}[t]
    
    \centering
    \subfloat[$D_x=D_y=10$ m]{\includegraphics[width=\linewidth]{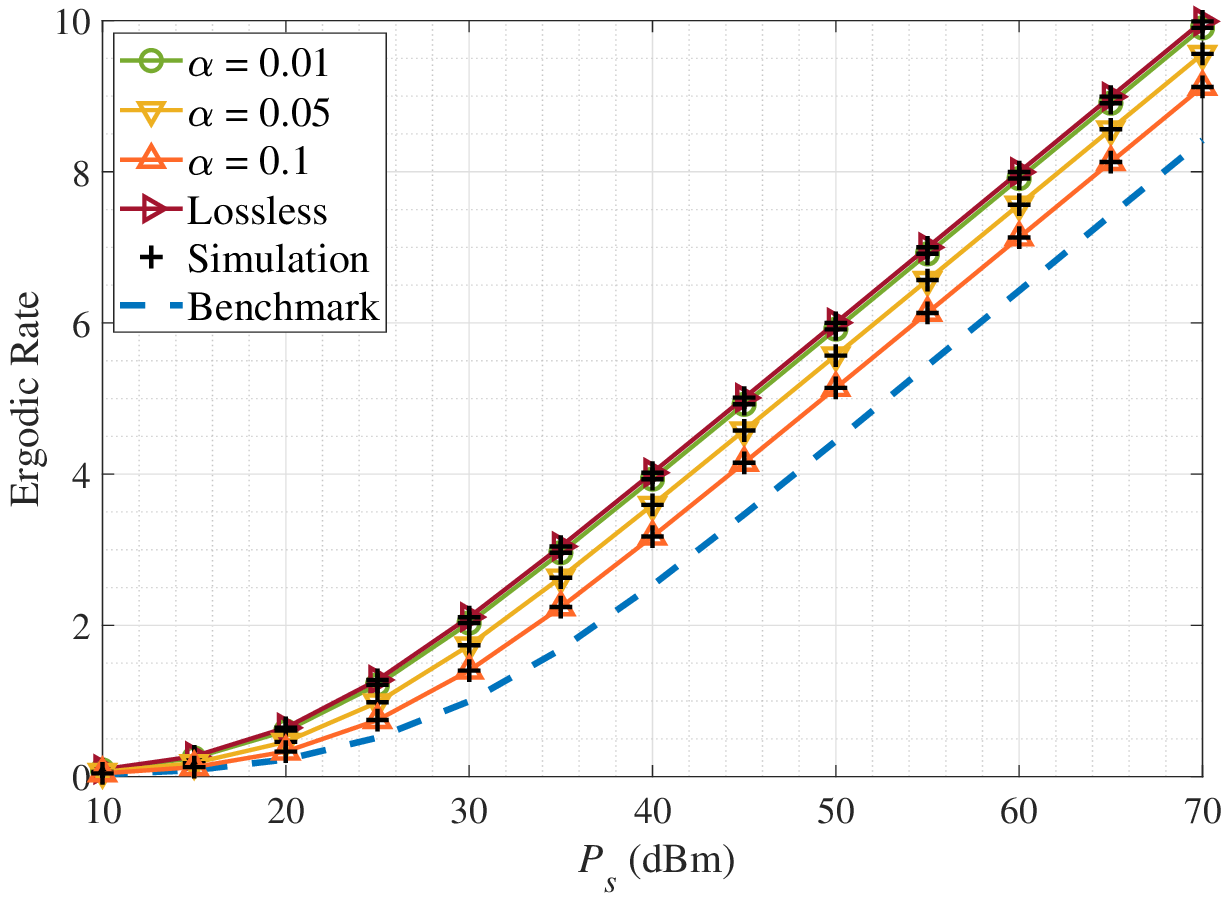}}\\
    \subfloat[$D_x=D_y=30$ m]{\includegraphics[width=\linewidth]{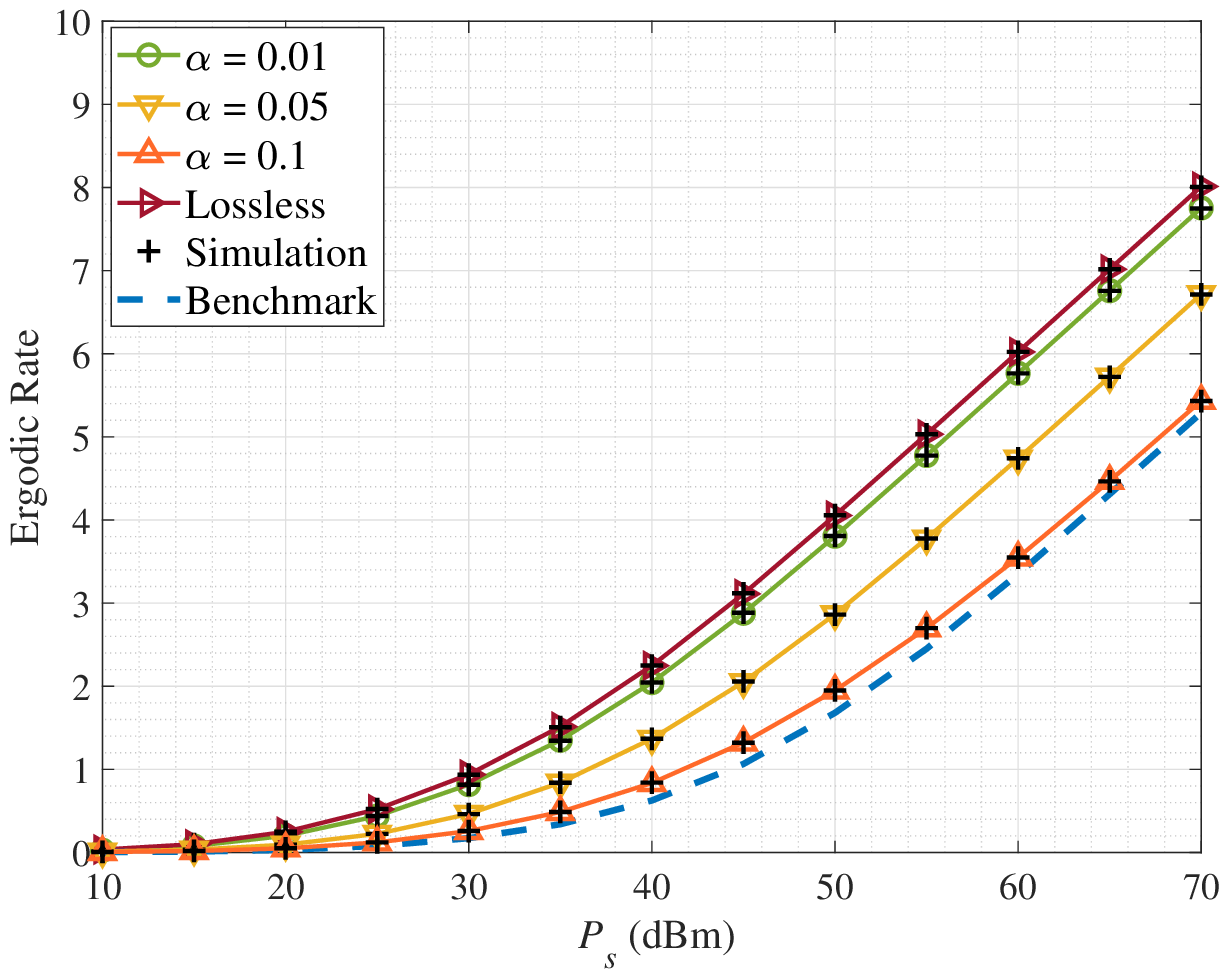}}
    
    \caption{Ergodic rate versus $P_s$ for different $\alpha$, where (a) $D_x=D_y=10$ m, and (b) $D_x=D_y=30$ m.}
    \label{Avg_rate_Ps_alpha}
\end{figure}

Fig. \ref{Avg_rate_Ps_DxDy} shows that the ergodic rate versus the transmission power $P_s$ for different dimensions of user area $D_x$ and $D_y$.
One can observe from the figure that the closed-form expressions of ergodic rate closely aligns with the simulation results, validating the accuracy of the theoretical analysis of the ergodic rate.
Intuitively, as $P_s$ increases, the ergodic rate increases steadily, while the increased $D_x$ and $D_y$ lead to the decrease of ergodic rate.
Additionally, under the same $P_s$, the difference in ergodic rate between wireless-powered PAS and traditional WPC system increases as $D_x$ and $D_y$ increase, pointing out that PA can effectively offset the adverse effects of an enlarged user area.
In contrast to low $D_x$ and $D_y$, under high $D_x$ and $D_y$, the difference in ergodic rate between wireless-powered PAS and traditional WPC system significantly rises as $P_s$ increases.
This result is in light of the fact that the high transmission loss caused by a wider user area can notably reduce the ergodic rate of the system, while the increase in transmission power further enhances the PA's capability to mitigate the path loss, enabling wireless-powered PAS to achieve a higher ergodic rate than traditional WPC system.
This conclusion highlights that the range of the user area can be flexibly controlled according to actual requirements to obtain the desired results.

Fig. \ref{Avg_rate_Ps_alpha} demonstrates that the ergodic rate versus the transmission power $P_s$ for different absorption coefficient $\alpha$, including the ideal case of a lossless waveguide.
Same as in Fig. \ref{OP_Ps_alpha}, Fig. \ref{Avg_rate_Ps_alpha}(a) and Fig. \ref{Avg_rate_Ps_alpha}(b) represent the cases of $D_x = D_y = 10$ m and $D_x = D_y = 30$ m, respectively, where $\alpha$ varies in the range of $0.01$ to $0.1$, and Fig. \ref{Avg_rate_Ps_alpha} draws a similar conclusion, namely, the larger $\alpha$, $D_x$, and $D_y$ are, the lower the ergodic rate of wireless-powered PAS is.
As expected, the ergodic rate grows with $P_s$, and the lossless waveguide case has the highest ergodic rate, showing the superiority of wireless-powered PAS in terms of data transmission.
Interestingly, by comparing Figs. \ref{OP_Ps_alpha}(b) and \ref{Avg_rate_Ps_alpha}(b), we find that when $\alpha=0.1$, the ergodic rate of wireless-powered PAS is better than that of traditional WPC system under high transmission power, but the conclusion regarding the outage probability is exactly the opposite.
The reason for this result is that the outage probability metric focuses on the area with low channel capacity, while the ergodic rate metric focuses on the ergodic achievable rate in the entire area.
Affected by high $\alpha$ and long waveguide, the in-waveguide propagation loss is quite large, resulting in a high outage probability for wireless-powered PAS when U is far away from PS and AP, while the ergodic rate of wireless-powered PAS is high when U is close to PS and AP, compensating for the high propagation loss and low rate when U is far away from PS and AP.

\begin{figure}[t]
    \centering
    \includegraphics[width=\linewidth]{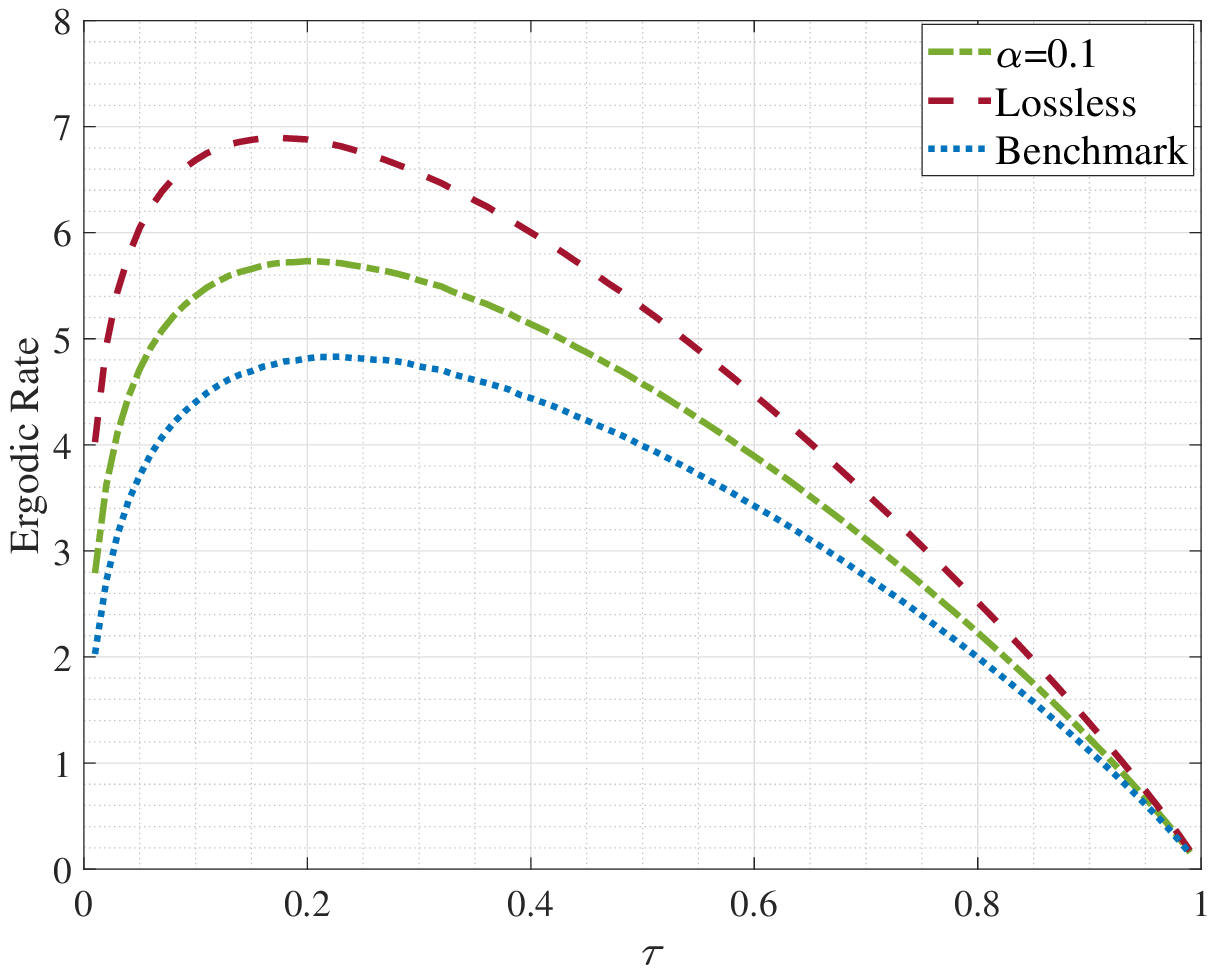}
    \caption{Ergodic rate versus $\tau$ for different $\alpha$, where $D_x=D_y=10$ m and $P_s=40$ dBm.}
    \label{Avg_rate_tau}
\end{figure}

Fig. \ref{Avg_rate_tau} depicts that the ergodic rate versus the time allocation factor $\tau$ for different transmission power $P_s$.
As shown in the figure, the ergodic rate of wireless-powered PAS first increases to its peak, and then rapidly declines to zero, suggesting that the wireless-powered PAS has an optimal time allocation factor $\tau^{opt}$ to maximize the ergodic rate.
Obviously, when $\tau$ approaches 1, the energy-constraint user has almost no time to harvest energy, so its energy is insufficient to support information transmission, leading to a decrease in the ergodic rate.
In contrast, when $\tau$ approaches 0, even if $\tau$ is very small, as long as $\tau$ is not equal to 0, the ergodic rate of wireless-powered PAS can still be maintained at a high level, indicating that wireless-powered PAS only needs a small amount of time for energy transfer to achieve a reliable communication.

\section{Conclusion}
In this paper, we first proposed the wireless-powered PAS, and derived the closed-form expressions of outage probability and ergodic rate for the practical lossy waveguide case and ideal lossless waveguide case, respectively, to comprehensively evaluate the communication performance of wireless-powered PAS.
Moreover, we achieved the optimal position analysis for the waveguides and user, providing valuable insights for guiding their deployments.

The results show that PA can be flexibly moved along the dielectric waveguide to mitigate the path loss and establish a strong LoS communication link.
In addition, the performance of wireless-powered PAS is severely affected by the absorption coefficient and the waveguide length, e.g., under conditions of high absorption coefficient and long waveguide, the outage probability of wireless-powered PAS is even worse than that of traditional WPC system.
However, the ergodic rate of wireless-powered PAS is better than that of traditional WPC system under conditions of high absorption coefficient and long waveguide.
Interestingly, the wireless-powered PAS has an optimal distance between PS and AP as well as an optimal time allocation factor to minimize the outage probability or maximize the ergodic rate.
Moreover, the system performance of PS and AP separated at the optimal distance between PS and AP is superior to that of PS and AP integrated into a HAP.


\renewcommand\thesubsectiondis{\Roman{subsection}.}

\setcounter{equation}{0}
\renewcommand{\theequation}{\thesection.\arabic{equation}}

\begin{appendices}
\section{Proof of Proposition 1}\label{AppendixA}

    Since $x_1^{pin} = x_2^{pin} = x_m$, the outage probability can be rewritten as
    \begin{align}\label{OP2}
    P_{out}=&{\rm Pr}\left[ \frac{\beta^2 \rho_t e^{-2\alpha x_m}}{ \left( \left(\frac{L}{2} - y_m \right)^2 + h^2 \right)  \left( \left(\frac{L}{2} + y_m \right)^2 + h^2 \right) } < \epsilon \right] \nonumber
    \\
    =&{\rm Pr}\left[ \frac{\chi e^{-2\alpha x_m}}{ \left( y_m^2 + h^2 - \frac{L^2}{4} \right)^2 + h^2 L^2 } < 1 \right] \nonumber
    \\
    =&{\rm Pr}\left[ \chi e^{-2\alpha x_m} - h^2 L^2  < \left( y_m^2 + h^2 - \frac{L^2}{4} \right)^2 \right].
    \end{align}
    
    To further derive $P_{out}$, we need to first determine the feasible region of $x_m$ and $y_m$, and then integrate over $x_m$ and $y_m$, respectively.
    Thus, the derivation of $P_{out}$ from the above inequality can be divided into four cases.

    \textbf{Case 1:} If $\chi e^{-2\alpha x_m} - h^2 L^2 \leqslant 0$, the outage probability of case 1 is given by
    \begin{equation}\label{case1}
        P_{out}^{\Rmnum{1}} = {\rm Pr}\left[ D_x \geqslant x_m \geqslant a \right] = \left\{
            \begin{aligned}
            &1,\quad a \leqslant  0\\
            &1-\frac{a}{D_x}, \quad a > 0\\
            &0, \quad a\geqslant D_x\\
            \end{aligned}
        \right.
    \end{equation}
    where $a \leqslant 0$, $a > 0$, and $a \geqslant D_x$ denote $h^2 L^2 \geqslant \chi$, $h^2 L^2 < \chi$, and $ h^2 L^2 \leqslant \chi e^{-2\alpha D_x}$, respectively.

    \textbf{Case 2:} If $\chi e^{-2\alpha x_m} - h^2 L^2 > 0$, $ h^2 - \frac{L^2}{4} > 0$ and $\sqrt{\chi e^{-2\alpha x_m} - h^2 L^2} + \frac{L^2}{4} - h^2 < 0$, the outage probability of case 2 is given by
    \begin{equation}\label{case2}
        P_{out}^{\Rmnum{2}} = {\rm Pr}\left[ a > x_m > b \right] = \left\{
            \begin{aligned}
            &1,\quad a \geqslant  D_x \ {\rm and} \ b \leqslant 0\\
            &1-\frac{b}{D_x}, \quad a\geqslant D_x \ {\rm and} \ b>0 \\
            &\frac{a}{D_x}, \quad  D_x > a > 0 \ {\rm and} \ b \leqslant 0 \\
            &\frac{a-b}{D_x}, \quad  D_x > a > b > 0 \\
            &0, \quad a \leqslant  0 \ {\rm or} \ b \geqslant D_x\\
            \end{aligned}
        \right.
    \end{equation}
    where $a<D_x$, $b \leqslant 0$, $b>0$, $b<D_x$, and $b \geqslant D_x$ denote $h^2 L^2 > \chi e^{-2\alpha D_x}$, $\left(h^2 - \frac{L^2}{4} \right)^2 + h^2L^2 \geqslant \chi$, $\left(h^2 - \frac{L^2}{4} \right)^2 + h^2L^2 < \chi$, $\left(h^2 - \frac{L^2}{4} \right)^2 + h^2L^2 >  \chi e^{-2 \alpha D_x}$, and $\left(h^2 - \frac{L^2}{4} \right)^2 + h^2L^2 \leqslant   \chi e^{-2 \alpha D_x}$, respectively.
    
    \textbf{Case 3:} If $\chi e^{-2\alpha x_m} - h^2 L^2 > 0$, $ h^2 - \frac{L^2}{4} > 0$ and $ \frac{D_y^2}{4} \geqslant  \sqrt{\chi e^{-2\alpha x_m} - h^2 L^2} + \frac{L^2}{4} - h^2 \geqslant  0$, we denote $z=\left( y_m^2 + h^2 - \frac{L^2}{4} \right)^2 + h^2 L^2$.
    Before deriving the outage probability of case 3, we need to obtain the cumulative distribution function (CDF) of $z$.
    \begin{align}\label{CDF}
        F_Z(z) =& {\rm Pr}\left[ \left( y_m^2 + h^2 - \frac{L^2}{4} \right)^2 + h^2 L^2 < z \right] \nonumber
        \\
        =&{\rm Pr}\left[  y_m^2 <  \sqrt{z - h^2 L^2} + \frac{L^2}{4} - h^2  \right] \nonumber
        \\
        =&\int_{-\sqrt{ \sqrt{z - h^2 L^2} + \frac{L^2}{4} - h^2 }}^{\sqrt{ \sqrt{z - h^2 L^2} + \frac{L^2}{4} - h^2 }} f(y_m) \,dy_m \nonumber
        \\
        =&\frac{2}{D_y}\sqrt{ \sqrt{z - h^2 L^2} + \frac{L^2}{4} - h^2 }.
    \end{align}

    Utilizing (\ref{CDF}), the outage probability of case 3 is given by
    \begin{align}\label{P1_case3_1}
        P_{out}^{\Rmnum{3}} =& {\rm Pr}\left[ z > \chi e^{-2\alpha x_m} \right] \nonumber
        \\
        =& \int_{x_1}^{x_2} \left( 1 - F_Z \left(\chi e^{-2\alpha x_m} \right) \right) f(x_m)  \,dx_m \nonumber
        \\
        =& \int_{x_1}^{x_2} \bigg( 1 - \frac{2}{D_x D_y} \nonumber
        \\
        &\times \sqrt{ \sqrt{\chi e^{-2\alpha x_m} - h^2 L^2} + \frac{L^2}{4} - h^2 } \bigg) \,dx_m.
    \end{align}
    
    Since $ \frac{D_y^2}{4} \geqslant  \sqrt{\chi e^{-2\alpha x_m} - h^2 L^2} + \frac{L^2}{4} - h^2 \geqslant  0$, i.e., $b \geqslant x_m \geqslant c$, we have
    \begin{equation}\label{P1_case3_2}
        \left\{
            \begin{aligned}
            &x_1 = 0 \quad x_2 = b, \quad  D_x > b > 0 \ {\rm and} \ c \leqslant 0 \\
            &x_1 = c \quad x_2 = b, \quad  D_x > b > c > 0 \\
            &x_1 = 0 \quad x_2 = D_x,\quad b \geqslant  D_x \ {\rm and} \ c \leqslant 0\\
            &x_1 = c \quad x_2 = D_x, \quad b\geqslant D_x \ {\rm and} \ D_x > c >0 \\
            &P_{out}^{\Rmnum{3}} = 0, \quad b \leqslant  0 \ {\rm or} \ c \geqslant D_x\\
            \end{aligned}
        \right.
    \end{equation}
    where $c \leqslant 0$, $c>0$, $c<D_x$, and $c \geqslant D_x$ denote $\left( \frac{D_y^2}{4} + h^2 - \frac{L^2}{4} \right)^2 + h^2L^2 \geqslant \chi$, $\left( \frac{D_y^2}{4} + h^2 - \frac{L^2}{4} \right)^2 + h^2L^2 < \chi$, $\left( \frac{D_y^2}{4} + h^2 - \frac{L^2}{4} \right)^2 + h^2L^2 >  \chi e^{-2 \alpha D_x}$, and $\left( \frac{D_y^2}{4} + h^2 - \frac{L^2}{4} \right)^2 + h^2L^2 \leqslant   \chi e^{-2 \alpha D_x}$, respectively.

    However, substituting (\ref{P1_case3_2}) into (\ref{P1_case3_1}), it is rather challenging to derive the closed-form expressions for $P_{out}^{\Rmnum{3}}$ due to its complicated integral.
    To overcome this problem, we first convert the interval $\left[x_1,x_2\right]$ to $\left[-1,1\right]$, and then use an efficient Gaussian-Chebyshev quadrature \cite[Eq. (25.4.38)]{Abramowitz1972Handbook} to obtain a close approximation, as shown in (\ref{P1_case3_3}).

    \begin{figure*}[!t]
    \begin{equation}\label{P1_case3_3}
        P_{out}^{\Rmnum{3}} =\left\{
            \begin{aligned}
            &\frac{b}{D_x}\left( 1 - \frac{\pi}{D_y K} \sum_{k = 1}^{K} \sqrt{1-\varpi_k^2} \sqrt{ \sqrt{\chi e^{-2\alpha \varphi_{k,1}} - h^2 L^2} + \frac{L^2}{4} - h^2} \right), \quad  D_x > b > 0 \ {\rm and} \ c \leqslant 0 \\
            &\frac{b-c}{D_x}\left( 1 - \frac{\pi}{D_y K} \sum_{k = 1}^{K} \sqrt{1-\varpi_k^2} \sqrt{ \sqrt{\chi e^{-2\alpha \varphi_{k,2}} - h^2 L^2} + \frac{L^2}{4} - h^2} \right), \quad  D_x > b > c > 0 \\
            &1 - \frac{\pi}{D_y K} \sum_{k = 1}^{K} \sqrt{1-\varpi_k^2} \sqrt{ \sqrt{\chi e^{-2\alpha \varphi_{k,3}} - h^2 L^2} + \frac{L^2}{4} - h^2},\quad b \geqslant  D_x \ {\rm and} \ c \leqslant 0\\
            &\frac{D_x-c}{D_x}\left( 1 - \frac{\pi}{D_y K} \sum_{k = 1}^{K} \sqrt{1-\varpi_k^2} \sqrt{ \sqrt{\chi e^{-2\alpha \varphi_{k,4}} - h^2 L^2} + \frac{L^2}{4} - h^2} \right), \quad b\geqslant D_x \ {\rm and} \ D_x > c >0 \\
            &0, \quad b \leqslant  0 \ {\rm or} \ c \geqslant D_x\\
            \end{aligned}
        \right.
    \end{equation}
    \hrulefill
    \end{figure*}

    \textbf{Case 4:} If $ \sqrt{\chi e^{-2\alpha x_m} - h^2 L^2} + \frac{L^2}{4} - h^2 > \frac{D_y^2}{4}$ and $\chi e^{-2\alpha x_m} - h^2 L^2 > 0$, $P_{out}^{\Rmnum{4}} = 0$.

    Considering the distribution of the four cases on the interval $[0, Dx]$ and $P_{out} = P_{out}^{\Rmnum{1}} + P_{out}^{\Rmnum{2}} + P_{out}^{\Rmnum{3}} + P_{out}^{\Rmnum{4}}$, the final results of Table \ref{table1} is derived, and the proof of Proposition \ref{proposition1} is completed.

\setcounter{equation}{0}
\section{Proof of Proposition 2}\label{AppendixB}

Substituting (\ref{SNRa}) into (\ref{Avg}), and considering $x_1^{pin} = x_2^{pin} = x_m$, the ergodic rate can be rewritten as
\begin{align}\label{Avg2}
    R_p&= (1-\tau) \mathbb{E} \left[ \log_2 \left( 1 + \frac{\beta^2 \rho_t e^{-2\alpha x_m}}{ \left( y_m^2 + h^2 - \frac{L^2}{4} \right)^2 + h^2 L^2 } \right) \right] \nonumber
    \\
    =& \frac{(1-\tau)}{D_x D_y \ln2}  \int_{-\frac{D_y}{2}}^{\frac{D_y}{2}} \int_{0}^{D_x} \ln \left( 1 + \frac{\beta^2 \rho_t e^{-2\alpha x_m}}{ p(y_m) } \right) dx_m dy_m   \nonumber
    \\
    \overset{(a)}{=}& \frac{2(1-\tau)}{D_x D_y \ln2} I_A - \frac{2(1-\tau)}{D_y \ln2} I_B,
\end{align}
where $I_A = \int_{0}^{\frac{D_y}{2}} \int_{0}^{D_x} \ln \left( p(y_m) + \beta^2 \rho_t e^{-2\alpha x_m} \right) dx_m dy_m   $, $I_B = \int_{0}^{\frac{D_y}{2}} \ln \left( p(y_m) \right)  dy_m$, and step $(a)$ uses the fact that $p(y_m) = ( y_m^2 + h^2 - \frac{L^2}{4} )^2  + h^2 L^2$ is an even function with respect to $y_m$. 

Since $p(y_m)$ is independent of $x_m$, we simply replace $p(y_m)$ with $p$ when solving the integral with respect to $x_m$.
Therefore, by setting $t = e^{-\alpha x_m}$, $I_A$ can be rewritten as
\begin{align}\label{IA}
    I_A = \frac{1}{\alpha}  \int_{0}^{\frac{D_y}{2}} \underbrace{\int_{e^{-\alpha D_x}}^{1} \frac{ \ln \left( p + \beta^2 \rho_t t^2 \right) }{t}  dt}_{I_A'} dy_m.
\end{align}

Let $q = \frac{p}{\beta^2 \rho_t}$, and $u = t^2$, $I_A'$ is given by
\begin{align}\label{IA'}
    I_A' =&  \ln \left( \beta^2 \rho_t \right) \ln t \vert_{e^{-\alpha D_x}}^{1}   + \int_{e^{-\alpha D_x}}^{1} \frac{ \ln \left( t^2 + q \right) }{t}  dt \nonumber
    \\
    =&\alpha D_x \ln \left( \beta^2 \rho_t \right) + \frac{1}{2} \ln q \ln u \vert_{e^{-2\alpha D_x}}^{1} \nonumber
    \\
    &+ \frac{1}{2} \underbrace{\int_{e^{-2\alpha D_x}}^{1} \frac{ \ln ( 1 + \frac{u}{q} ) }{u} du}_{I_A''}.
\end{align}

Let $k=\frac{u}{q}$ and $m = -k$, $I_A''$ is expressed as 
\begin{align}\label{IA''}
    I_A'' =&  \int_{\frac{e^{-2\alpha D_x}}{q}}^{\frac{1}{q}} \frac{ \ln ( 1 + k ) }{k} dk \nonumber
    \\
    =&\int_{-\frac{e^{-2\alpha D_x}}{q}}^{-\frac{1}{q}} \frac{ \ln ( 1 - m ) }{m} dm \nonumber
    \\
    =&{\rm Li}_2\left( -\frac{e^{-2\alpha D_x}}{q} \right) - {\rm Li}_2\left( -\frac{1}{q} \right),
\end{align}
where ${\rm Li}_2\left( x \right) = -\int_{0}^{x} \frac{1-m}{m} \,dm $ \cite[Eq. (6.254.1)]{10.1115/1.3138251}.

Substituting (\ref{IA'}) and (\ref{IA''}) into (\ref{IA}), we have
\begin{align}\label{IA_result}
    I_A =& \frac{1}{\alpha}  \int_{0}^{\frac{D_y}{2}} \alpha D_x \ln p(y_m) + \frac{1}{2} \bigg(  {\rm Li}_2\left( -\frac{\beta^2 \rho_t e^{-2\alpha D_x}}{p(y_m)} \right) \nonumber
    \\
    &- {\rm Li}_2\left( -\frac{\beta^2 \rho_t}{p(y_m)} \right)   \bigg) dy_m \nonumber
    \\
    =&D_x I_B + \frac{1}{2\alpha} \int_{0}^{\frac{D_y}{2}}  {\rm Li}_2\left( -\frac{\beta^2 \rho_t e^{-2\alpha D_x}}{p(y_m)} \right) \nonumber
    \\
    &- {\rm Li}_2\left( -\frac{\beta^2 \rho_t}{p(y_m)} \right)    dy_m.
\end{align}

Inserting (\ref{IA_result}) into (\ref{Avg2}), $I_B$ is removed, and exploiting the Gaussian-Chebyshev quadrature, the final result of (\ref{Avg_result}) is derived, which concludes the proof of Proposition \ref{proposition2}.

\setcounter{equation}{0}
\section{Proof of Proposition 3}\label{AppendixC}

Since $\alpha=0$, and $x_1^{pin} = x_2^{pin} = x_m$, the outage probability can be rewritten as
\begin{align}\label{OP3}
    P_{out}=&{\rm Pr}\left[ \frac{\beta^2 \rho_t }{ \left( \left(\frac{L}{2} - y_m \right)^2 + h^2 \right)  \left( \left(\frac{L}{2} + y_m \right)^2 + h^2 \right) } < \epsilon \right] \nonumber
    \\
    =&{\rm Pr}\left[ \frac{\chi }{ \left( y_m^2 + h^2 - \frac{L^2}{4} \right)^2 + h^2 L^2 } < 1 \right] \nonumber
    \\
    =&{\rm Pr}\left[ \chi  - h^2 L^2  < \left( y_m^2 + h^2 - \frac{L^2}{4} \right)^2 \right].
\end{align}

Similar to Proposition \ref{proposition1}, the derivation of $P_{out}$ is also divided into four cases.

\textbf{Case 1:} If $\chi - h^2 L^2 \leqslant 0$, $P_{out}^{\Rmnum{1}} = 1$.

\textbf{Case 2:} If $\chi  - h^2 L^2 > 0$, $ h^2 - \frac{L^2}{4} > 0$ and $\sqrt{\chi - h^2 L^2} + \frac{L^2}{4} - h^2 < 0$, $P_{out}^{\Rmnum{2}} = 1$.

\textbf{Case 3:} If $\chi - h^2 L^2 > 0$, $ h^2 - \frac{L^2}{4} > 0$ and $ \frac{D_y^2}{4} \geqslant  \sqrt{\chi - h^2 L^2} + \frac{L^2}{4} - h^2 \geqslant  0$, we have
\begin{align}\label{P3_case3}
    P_{out}^{\Rmnum{3}}=&{\rm Pr}\left[ z > \chi  \right] \nonumber
    \\
    =&1 - F_Z(\chi) \nonumber
    \\
    =&1 - \frac{2}{D_y} \sqrt{ \sqrt{\chi - h^2 L^2} + \frac{L^2}{4} - h^2}.
\end{align}

\textbf{Case 4:} If $ \sqrt{\chi - h^2 L^2} + \frac{L^2}{4} - h^2 > \frac{D_y^2}{4}$ and $\chi - h^2 L^2 > 0$, $P_{out}^{\Rmnum{4}} = 0$.

Similarly, $P_{out} = P_{out}^{\Rmnum{1}} + P_{out}^{\Rmnum{2}} + P_{out}^{\Rmnum{3}} + P_{out}^{\Rmnum{4}}$, the final results of Table \ref{table2} is derived, and the proof of Proposition \ref{proposition3} is completed.

\setcounter{equation}{0}
\section{Proof of Proposition 4}\label{AppendixD}

Considering $\alpha=0$, and $x_1^{pin} = x_2^{pin} = x_m$, the ergodic rate can be rewritten as
\begin{align}\label{Avg3}
    R_p&= (1-\tau) \mathbb{E} \left[ \log_2 \left( 1 + \frac{\beta^2 \rho_t }{ \left( y_m^2 + h^2 - \frac{L^2}{4} \right)^2 + h^2 L^2 } \right) \right] \nonumber
    \\
    =& \frac{2(1-\tau)}{D_y \ln2}  \int_{0}^{\frac{D_y}{2}} \ln \left( 1 + \frac{\beta^2 \rho_t }{  \left( y_m^2 + h^2 - \frac{L^2}{4} \right)^2 + h^2 L^2 } \right) dy_m   \nonumber
    \\
    =& \frac{2(1-\tau)}{D_y \ln2} \left( I_C -  I_B  \right),
\end{align}
where $I_C =  \int_{0}^{\frac{D_y}{2}} \ln \Big( \left( y_m^2 + h^2 - \frac{L^2}{4} \right)^2 + h^2 L^2 + \beta^2 \rho_t  \Big) dy_m$.

By using factorization, $I_B$ is expressed as
\begin{align}\label{IB}
    I_B =& \underbrace{\int_{0}^{\frac{D_y}{2}} \ln \left( y_m^2 + L y_m + h^2 + \frac{L^2}{4} \right)  dy_m}_{I_B'} \nonumber
    \\
    &+\underbrace{\int_{0}^{\frac{D_y}{2}} \ln \left( y_m^2 - L y_m + h^2 + \frac{L^2}{4} \right)  dy_m}_{I_B''}.
\end{align}

Furthermore, performing integration by parts, $I_B'$ is given by
\begin{align}\label{IB'}
    I_B' =& \frac{D_y}{2} \ln \left( \left( \frac{D_y}{2} + \frac{L}{2} \right)^2 + h^2 \right) - D_y \nonumber
    \\
    & + \underbrace{\int_{0}^{\frac{D_y}{2}} \frac{Ly_m + 2h^2 + \frac{L^2}{2}}{ \left( y_m + \frac{L}{2} \right)^2 + h^2}  dy_m}_{I_B'''}.
\end{align}

Let $w = y_m + \frac{L}{2}$, we obtain
\begin{align}\label{IB'''}
    I_B''' =& \int_{\frac{L}{2}}^{\frac{D_y + L}{2}} \frac{Lw + 2h^2}{ w^2 + h^2}  dw \nonumber
    \\
    =&L\int_{\frac{L}{2}}^{\frac{D_y + L}{2}} \frac{w}{ w^2 + h^2} dw + 2h^2 \int_{\frac{L}{2}}^{\frac{D_y + L}{2}} \frac{1}{ w^2 + h^2}  dw \nonumber
    \\
    =&\frac{L}{2} \ln \left( \frac{ (D_y + L)^2 + 4h^2 }{L^2 + 4h^2} \right) + 2h \Big( \arctan \left( \frac{D_y + L}{2h} \right) \nonumber
    \\
    &- \arctan \left( \frac{L}{2h} \right) \Big).
\end{align}

Substituting (\ref{IB'''}) into (\ref{IB'}), $I_B'$ is expressed as
\begin{align}\label{IB'_result}
    I_B' =& \frac{D_y + L}{2} \ln \left( \left( D_y + L \right)^2 + 4h^2 \right) - (\ln 2 + 1)D_y  \nonumber
    \\
    &- \frac{L}{2} \ln (L^2 + 4h^2) + 2h \Big( \arctan \left( \frac{D_y + L}{2h} \right) \nonumber
    \\
    &- \arctan \left( \frac{L}{2h} \right) \Big).
\end{align}

In addition, by replacing $L$ with $-L$, $I_B''$ is written as
\begin{align}\label{IB''_result}
    I_B'' =& \frac{D_y-L}{2} \ln \left( \left( D_y - L \right)^2 + 4h^2 \right) - (\ln 2 + 1)D_y  \nonumber
    \\
    &+ \frac{L}{2} \ln (L^2 + 4h^2) + 2h \Big( \arctan \left( \frac{D_y - L}{2h} \right) \nonumber
    \\
    &+ \arctan \left( \frac{L}{2h} \right) \Big).
\end{align}

Inserting (\ref{IB'_result}) and (\ref{IB''_result}) into (\ref{IB}), $I_B$ is derived.

By using factorization, $v=\sqrt{2\left( u - h^2 + \frac{L^2}{4} \right)}$, and $u = \sqrt{ \left( h^2 - \frac{L^2}{4} \right)^2 + h^2L^2 + \beta^2 \rho_t }$, $I_C$ is expressed as
\begin{align}\label{IC}
    I_C =&  \underbrace{\int_{0}^{\frac{D_y}{2}} \ln \left( y_m^2 + v y_m + u \right)  dy_m }_{I_C'}\nonumber
    \\
    &+ \underbrace{\int_{0}^{\frac{D_y}{2}} \ln \left( y_m^2 - v y_m + u \right)  dy_m}_{I_C''}.
\end{align}

Following the similar steps to the derivation of $I_B'$, and using $r = y_m + \frac{p}{2}$ and $s = \sqrt{u - \frac{v^2}{4}}$, we have
\begin{align}\label{IC'_result}
    I_C' =& \frac{D_y + v}{2} \ln \left( \left( D_y + v \right)^2 + 4s^2 \right) - (\ln 2 + 1)D_y  \nonumber
    \\
    &- \frac{v}{2} \ln (4u) + 2s \Big( \arctan \left( \frac{D_y + v}{2s} \right) \nonumber
    \\
    &- \arctan \left( \frac{v}{2s} \right) \Big).
\end{align}

Moreover, replacing $v$ in (\ref{IC'_result}) with $-v$, $I_C''$ is given by
\begin{align}\label{IC''_result}
    I_C'' =& \frac{D_y - v}{2} \ln \left( \left( D_y - v \right)^2 + 4s^2 \right) - (\ln 2 + 1)D_y  \nonumber
    \\
    &+ \frac{v}{2} \ln (4u) + 2s \Big( \arctan \left( \frac{D_y - v}{2s} \right) \nonumber
    \\
    &+ \arctan \left( \frac{v}{2s} \right) \Big).
\end{align}
Thus, $I_C = I_C' + I_C''$ is obtained.

Finally, substituting $I_B$ and $I_C$ into (\ref{Avg3}), we obtain the result of (\ref{Avg_result_lossless}).
The proof of Proposition \ref{proposition4} is completed.

\end{appendices}

\bibliographystyle{IEEEtran}
\bibliography{IEEEabrv,citation}

\end{document}